\documentclass[english,11pt,a4paper,onecolumn,accepted=2026-03-02]{quantumarticle}
\pdfoutput=1
\usepackage[english]{babel}
\usepackage{latexsym,amssymb,amsmath,amsfonts,amscd,epsfig,amsthm,color,mathrsfs, mathtools}
\usepackage{graphicx}
\usepackage{float}
\usepackage[left=2.3cm,right=2.3cm,top=2.3cm,bottom=2.3cm]{geometry}
\usepackage{authblk}
\usepackage{colonequals}
\usepackage{array}
\usepackage{authblk}

\usepackage{tikz}
\usetikzlibrary{quantikz2, positioning}

\usepackage{belov_main}
\usepackage[persection, figure]{belov_paper}

\newcommand{\OO}{\mathcal{O}}
\newcommand{\tOO}{\widetilde{\OO}}

\usepackage{mathabx}
\newcommand{\maps}[1]{\;\stackrel{#1}{\longmapsto}\;}
\newcommand{\transduce}[1]{\stackrel{#1}{\rightsquigarrow}}
\newcommand{\DownTransduce}{\scalebox{0.7}{\rotatebox[origin=c]{270}{$\rightsquigarrow$}}}

\usepackage{hyperref}[backref]
\hypersetup{
    pdftitle={Space-Efficient Quantum Error Reduction without log Factors}, % title
    pdfauthor={Aleksandrs Belovs and Stacey Jeffery}, % author
}

\begin{document}

\title{Space-Efficient Quantum Error Reduction without log Factors}
\author{Aleksandrs Belovs}
\affiliation{Center for Quantum Computing Science, Faculty of Science and Technology, University of Latvia}
\author{Stacey Jeffery}
\affiliation{QuSoft, CWI \& University of Amsterdam, the Netherlands}

\maketitle

\mycutecommand{\q}{{\mathrm{q}}}
\mycutecommand{\w}{{\mathrm{w}}}

\definecolor{applegreen}{rgb}{0.55, 0.71, 0.0}

\def\xicolor{red}
\def\taucolor{applegreen}
\def\witnesscolor{blue}
\def\nonquerycolor{gray}
\def\querycolor{orange}

\begin{abstract}
    Given an algorithm that outputs the correct answer with bounded error, say $1/3$, it is sometimes desirable to reduce this error to some arbitrarily small $\eps$ -- for example, if one wants to call the algorithm many times as a subroutine. The usual method, for both quantum and randomized algorithms, is a procedure called majority voting, which incurs a multiplicative overhead of $\OO(\log\frac{1}{\eps})$ from calling the algorithm this many times. 

    A recent paper~\cite{belovs:taming} 
%[Quantum 2024] 
    introduced a model of quantum computation called \emph{transducers}, and showed how to reduce the ``error'' of a transducer arbitrarily with only constant overhead, using a construction analogous to majority voting called \emph{purification}. 
    Even error-free transducers map to bounded-error quantum algorithms, so this does not let you reduce algorithmic error for free, but it does allow bounded-error quantum algorithms to be composed without incurring log factors.

In this paper, we present a new highly simplified construction of a purifier, that can be understood as a weighted walk on a line similar to a random walk interpretation of majority voting.
Our purifier has much smaller space and time complexity than the previous one.
Indeed, it only uses one additional counter, and only performs two increment and two decrement operations on each iteration.
It also has quadratically better dependence on the soundness-completeness gap of the algorithm being purified.
We prove that our purifier has optimal query complexity, even down to the constant, which matters when one composes quantum algorithms to super-constant depth.

Purifiers can be seen as a way of turning a ``Monte Carlo'' quantum algorithm into a ``Las Vegas'' quantum algorithm -- a process for which there is no classical analogue -- leading to phenomena like the above-mentioned composition of bounded-error quantum algorithms without log factors. Conceptually, our simplified construction sheds light on this strange quantum phenomenon, and, more practically, could have implications for the time and space complexity of composed quantum algorithms.

\end{abstract}

\section{Introduction}

Error reduction
is one of the most fundamental algorithmic primitives in both the randomized and quantum settings.
It transforms a program that evaluates a function $f$ with bounded error, say $1/3$, into a program that evaluates the same function with much smaller error $\eps$.
In both the quantum and randomized cases, the conventional procedure is based on majority voting.

If we want to reduce the error of a complete quantum algorithm that ends with a measurement outputting a classical string, then we can treat it as a randomized process and apply majority voting just as we would in a randomized algorithm: execute the quantum algorithm together with the final measurement a number of times, and output the majority. %

Somewhat more interesting is error reduction of a \emph{coherent} quantum subroutine. 
If the subroutine is executed several times, such as when it is employed as an input oracle for another quantum algorithm, then its error must be driven down in order to keep the total error of the whole algorithm sufficiently small.%
\footnote{If the best upper bound we can prove on an algorithm's error probability is $\dfrac 12$ or larger, the algorithm is widely considered to be useless.}
As quantum subroutines may be called in superposition, the error-reduced subroutine cannot use measurements, but must instead be a coherent subroutine. 

Let us formulate the settings more formally.
We will assume the binary case for most of the paper.
This is without loss of generality, as the general case can be obtained from the binary case using the Bernstein-Vazirani trick as we explain in \rf{sec:non-Boolean}.
We treat the subroutine whose error has to be reduced as an input oracle $O$ to the error reduction procedure.
The most natural choice is to take the following \emph{state-generating oracle}
\begin{equation}
\label{eqn:InputOracle}
O\colon \ket A|0>\ket W|0> \longmapsto \ket A\reg W |\phi> = \sqrt{1-p} \ket A |0> \ket W|\phi_0> + \sqrt{p} \ket A|1>\ket W|\phi_1>,
\end{equation}
where the first register $\reg A$ is the output qubit (answer register), and the second register $\reg W$ contains some internal state of the subroutine (workspace), and both states $\phi_0$ and $\phi_1$ are normalised.
The error reduction procedure $R$, given this input oracle, must perform the transformation
\begin{equation}
\label{eqn:ErrorReduction}
R(O)\colon \ket A|0> \longmapsto
\begin{cases}
\ket A|0>, &\text{if $p\le \tfrac 12 - \delta$;}\\
\ket A|1>, &\text{if $p\ge \tfrac 12 + \delta$;}
\end{cases}
\end{equation}
with sufficiently high accuracy.
Here $2 \delta > 0$ is the gap between the accepting and the rejecting probabilities of the subroutine $O$. 
That is, the error of the subroutine is $1/2-\delta$.
The space complexity of the subroutine is the number of qubits used by $\cA$ and $\cW$. %

In the remaining part of the introduction, we first recall the traditional approach to error reduction, that literally performs the transformation in~\rf{eqn:ErrorReduction}.
Then, we consider approaches that improve on the traditional approach by avoiding the transformation in~\rf{eqn:ErrorReduction}.
They are based on idealised models of computation like span programs, dual adversary, and transducers.

\subsection{Traditional Approach to Error Reduction}

The most well-known approach to quantum error reduction is still majority voting. 

\begin{thm}[Quantum Majority Voting, folklore]
\label{thm:majorityVoting}
Let $\eps, \delta>0$ be real parameters.
Quantum majority voting $\eps$-approximately performs the transformation~\rf{eqn:ErrorReduction} for any input oracle of the form~\rf{eqn:InputOracle}.
It executes the input oracle $O$ and its inverse $O^*$ each $\ell = O\sA[\frac1{\delta^2}\log \frac 1\eps]$ times.
Its space complexity is $\ell s + \OO(\log \ell)$, where $s$ is the space complexity of $O$.
Besides executions of $O$ and $O^*$, it uses $\OO(\ell\log \ell)$ other 1- or 2-qubit operations.
\end{thm}

One difference between the classical and quantum case, is that we cannot na\"ively reuse the space used by the subroutine in the quantum case, but must instead make $\ell$ copies of the whole state $\ket{A\reg W}|\phi>$ (in contrast to classical majority voting, which can just keep track of how many 1s have been seen so far).
This gives the $\ell s$ term in the space complexity of \rf{thm:majorityVoting}.
Let us give a short proof in order to demonstrate this.

\mycutecommand{\Rhigh}{\cR_{\text{high}}}

\begin{proof}
Let $r$ be the required output of the algorithm in~\rf{eqn:ErrorReduction}, i.e., $r=0$ if $p\le\frac12-\delta$, and $r=1$ if $p\ge \frac12+\delta$.
We also assume that $\ell$ is a power of 2.
The algorithm works as follows.
The input oracle $O$ is executed $\ell$ times, each time on a new pair of registers $\reg A_i$, $\reg W_i$.
This gives a state of the form
\[
\sum_{b\in \bool^\ell} (1-p)^{\frac{\ell-|b|}2}p^{\frac{|b|}2} \ket A_1 |b_1>\cdots \ket A_\ell |b_\ell> \ket W_1 |\phi_{b_1}> \dots \ket W_\ell |\phi_{b_\ell}>,
\]
where $|b|$ is the Hamming weight of $b$.
Evaluation of the sum of the bits $b_1,\dots,b_\ell$ into a new register $\reg R$ gives the state 
\[
\ket|\psi> = \sum_{b\in \bool^\ell} (1-p)^{\frac{\ell-|b|}2}p^{\frac{|b|}2} \ket A_1 |b_1>\cdots \ket A_\ell |b_\ell> \ket W_1 |\phi_{b_1}> \dots \ket W_\ell |\phi_{b_\ell}> \ketA R||b|>.
\]
Let $\Rhigh$ be the most significant qubit of $\cR$.
We rewrite the state $\ket|\psi>$ as
\[
\ket|\psi> = \ket \Rhigh |r> \ket |\psi_r> + \ket \Rhigh |1-r> \ket|\psi_{1-r}>,
\]
where $\ket|\psi_r>$ and $\ket|\psi_{1-r}>$ are non-normalized states on the remaining registers.
Now we copy the qubit $\Rhigh$ into the output qubit $\reg A$, which gives the state
\begin{equation}
\label{eqn:majorityVotingIntermediateState1}
\ket A|r> \ket \Rhigh |r> \ket |\psi_r> + \ket A|1-r>\ket \Rhigh |1-r> \ket|\psi_{1-r}>,
\end{equation}
and run the whole preceding computation in reverse.

If, instead of~\rf{eqn:majorityVotingIntermediateState1}, we had the state
\begin{equation}
\label{eqn:majorityVotingIntermediateState2}
\ket A|r> \ket \Rhigh |r> \ket |\psi_r> + \ket A|r>\ket \Rhigh |1-r> \ket|\psi_{1-r}>,
\end{equation}
we would get the required state $\ket A|r>$ at the end exactly.
Since the difference between the states in~\rf{eqn:majorityVotingIntermediateState1} and~\rf{eqn:majorityVotingIntermediateState2} is of norm $\sqrt2 \|\psi_{1-r}\|$, the imprecision of the algorithm is also $\sqrt2 \|\psi_{1-r}\|$ by unitarity (see \rf{sec:perturbed}).

Let us estimate $\|\psi_{1-r}\|$.
We consider the case $r=0$ for concreteness, the second case being similar.
The norm squared of $\psi_{1-r}$ equals the probability of $|b|\ge \ell/2$ when $b$ is the sum of $\ell$ i.i.d. Bernoulli random variables with expectation $p$.
By Hoeffding's inequality:
\[
\|\psi_{1-r}\|^2
=
\Pr\skB[|b|\ge \ell/2]
\le
\Pr\skB[|b| - {\bE\skA[|b|]} \ge \ell\delta]
\le
\ee^{-2\ell\delta^2}
\le
\frac{\eps^2}{2}
\]
for some $\ell = \OO\sA[\frac1{\delta^2}\log \frac 1\eps]$.

The registers $\reg A_i$ and $\reg W_i$ use $\ell s$ qubits together, and $\cR$ uses additional $\OO(\log \ell)$.
Time complexity of the algorithm is easy to check.
\end{proof}

\mycutecommand{\Oref}{O_{\mathrm{ref}}}
Another approach uses quantum amplitude estimation.
For this primitive, it is sufficient to use a less powerful \emph{reflecting oracle}:
\begin{equation}
\label{eqn:Oref}
\Oref = O \mathrm{Ref_{\ket|0>\ket|0>}} O^*,
\end{equation}
where $\mathrm{Ref}_{\ket|0>\ket|0>}$ is the reflection about the state $\ket A|0>\ket W|0>$ in $\reg A\otimes \reg W$.
This oracle reflects about the state $\ket{}|\phi>$ in $\reg A\otimes \reg W$.
It is easy to implement $\Oref$ using two oracle calls and $s+\OO(1)$ additional operations.
The latter are usually majorized by the time complexity of implementing $O$, since a subroutine using $s$ qubits will generally have time complexity at least $s$.
On the other hand, implementing $O$ using $\Oref$ is complicated. (See~\cite{aaronson:counting, belovs:counting} for some lower bounds involving the reflecting oracle.)

Even more, it suffices for $\Oref$ to act as $O \mathrm{Ref_{\ket|0>\ket|0>}} O^*$ only on the span of the vectors $\ket A|0> \ket W|\phi_0>$ and $\ket A|1>\ket W|\phi_1>$, and it can be arbitrary on its orthogonal complement.
That is, we may assume that $\Oref$ satisfies the following weaker conditions:
\begin{equation}
\label{eqn:OrefGeneral}
\Oref\colon \begin{cases}
\sqrt{1-p} \ket|0> \ket|\phi_0> + \sqrt{p} \ket|1>\ket|\phi_1>
\mapsto
\sqrt{1-p} \ket|0> \ket|\phi_0> + \sqrt{p} \ket|1>\ket|\phi_1>\\
\sqrt{p} \ket |0>\ket|\phi_0> - \sqrt{1-p} \ket|1>\ket|\phi_1>
\mapsto
-\sqrt{p} \ket |0>\ket|\phi_0> + \sqrt{1-p} \ket|1>\ket|\phi_1>.
\end{cases}
\end{equation}

\begin{thm}[Quantum Amplitude Estimation~\cite{brassard:amplification}]
\label{thm:estimation}
Let $\eps, \delta>0$ be real parameters.
There exists a quantum algorithm that, given a copy of the state $\ket|\phi>$ as in~\rf{eqn:InputOracle} and query access to an oracle $\Oref$ satisfying~\rf{eqn:OrefGeneral}, outputs 0 if $p\le \tfrac12 - \delta$ and outputs 1 if $p\ge \tfrac12 + \delta$ except with error probability $\eps$.
It makes $\ell = \OO(\frac1{\delta\eps})$ (controlled) queries to the oracle $\Oref$,  has space complexity $s + \OO(\log \ell)$, and time complexity $\OO(\ell\log \ell)$.
\end{thm}
While amplitude estimation is typically stated in terms of calls to the state-generating oracle $O$ from~\rf{eqn:InputOracle}, it actually uses the reflecting oracle $\Oref$ defined as in \rf{eqn:Oref}, and it is simple to show that it is in fact sufficient for $\Oref$ to satisfy \rf{eqn:OrefGeneral}.

Majority voting has better dependence on $\eps$, and amplitude estimation on $\delta$.
It is possible to combine the two procedures to obtain the following result.

\begin{cor}[Quantum Error Reduction]\label{cor:ErrorReduction}
Let $\eps, \delta>0$ be parameters.
There exists a quantum algorithm that $\eps$-approximately performs the transformation~\rf{eqn:ErrorReduction} for any input oracle of the form~\rf{eqn:InputOracle}.
It executes the input oracle $O$ and its inverse $O^*$ both $\OO\sA[\frac1{\delta}\log \frac 1\eps]$ times.
Its space complexity is $\OO\sA[\log\frac 1\eps]\sA[s + \log \frac 1\delta]$, where $s$ is the space complexity of $O$.
Besides executions of $O$ and $O^*$, it uses $\tOO\sA[\frac1{\delta} \log\frac1\eps]$ other 1- or 2-qubit operations.
\end{cor}

\begin{proof}
We first construct an error reduction algorithm $V$ with constant error $\eps = 1/10$ using \rf{thm:estimation}.
It requires one copy of $\ket|\phi>$ and $\OO(1/\delta)$ executions of $\Oref$, which can be performed using $\OO(1/\delta)$ invocations of $O$ and $O^*$ due to~\rf{eqn:Oref}.
The space complexity of $V$ is $s + \OO\s[\log \frac1{\delta}]$.
Its time complexity is $\tOO\s[\frac1{\delta}]$.

To get the imprecision down to $\eps$, we use $V$ as an oracle in \rf{thm:majorityVoting} with $\delta$ being $\Omega(1)$.
This requires $O\s[\log \frac1\eps]$ executions of $V$ and $V^*$, and its space complexity is $O\s[\log \frac1\eps]$ times the space complexity of $V$ plus negligible $O\s[\log \frac1\eps]$.
Its time complexity is dominated by $O\s[\log \frac1\eps]$ times the time complexity of $V$.
\end{proof}

The number of queries made by this algorithm has optimal dependence on $\eps$ by~\cite{buhrman:boundsForSmallError} as well as on $\delta$ by~\cite{nayak:approximatingMedian}. 
However, this algorithm is not optimal in terms of memory usage.
It is possible to reuse the memory of the subroutine $O$ as shown by Marriott and Watrous in the context of QMA~\cite{marriott:QMA} -- see also~\cite{nagaj:fastQMA} -- but this construction still uses additional qubits.
Alternatively, one can use Quantum Signal Processing (QSP)~\cite{low:quantumSignalProcessing} for similar results~\cite{gilyen:quantumSingularValueTransformation}. 
For instance, we can get the following result, which, to the best of our knowledge, has not been previously stated explicitly.

\begin{thm}
[QSP-based Error Reduction]
\label{thm:ErrorReduction}
Let $\eps, \delta>0$ be real parameters.
There exists a quantum algorithm that, given query access to any oracle $\Oref$ satisfying~\rf{eqn:OrefGeneral}, $\eps$-approximately performs the following transformation
\begin{equation}
\label{eqn:ErrorReductionSign}
\ket |\phi> \longmapsto
\begin{cases}
\ket |\phi>, &\text{if $p\le \tfrac 12 - \delta$;}\\
-\ket |\phi>, &\text{if $p\ge \tfrac 12 + \delta$;}
\end{cases}
\end{equation}
for any normalised $\ket|\phi> \in \spn\sfigA{\ket|0>\ket|\phi_0>, \ket|1>\ket|\phi_1>}$ (including, but not limited to, the specific vector $\ket|\phi>$ from~\rf{eqn:InputOracle}).
The algorithm makes $\ell = \OO\sA[\frac1{\delta}\log\frac1{\eps}]$ queries to the oracle $\Oref$, and uses $O(\ell)$ other 1- or 2-qubit gates.
The algorithm uses no additional qubits compared to $\Oref$.
\end{thm}
We prove \rf{thm:ErrorReduction} in \rf{sec:QSP}. The proof is similar to a construction in~\cite[Section 3]{gilyen:quantumSingularValueTransformation}, used there for QMA amplification. See \rf{sec:QSP} for a more detailed comparison.

\subsection{Span Programs, Dual Adversary, Transducers, and Error Reduction}
\label{sec:spanPrograms}
Given the lower bounds from~\cite{buhrman:boundsForSmallError} and~\cite{nayak:approximatingMedian}, the algorithm in \rf{thm:ErrorReduction} seems to be the best possible technique for error reduction.
Also, often $\delta = \Omega(1)$, so the query complexity is only logarithmic, which might seem negligible. 
However, imagine a bounded-error subroutine that executes itself recursively many times.
In this case, error reduction for each level of recursion gives a multiplicative logarithmic factor, resulting in something like $\OO\s[\log \frac1\eps]^d$, where $d$ is the depth of recursion.
If $d$ is super-constant, this might be quite significant.

There exist several approaches to deal with error reduction by avoiding the transformation in~\rf{eqn:ErrorReduction} altogether.
They all use similar ideas, and we first explore the one based on span programs~\cite{reichardt:formulae}, which was directly inspired by quantum algorithms for the iterated NAND tree~\cite{farhi:nandTree, childs:NANDTree, ambainis:formulaeEvaluation}.
The span program model is an idealised computational model (meaning that it does not correspond to any actual computational device) for evaluating Boolean functions, with the following three key properties:

\newcommand{\PointRef}[1]{\hyperref[point#1]{Point~#1}}

\begin{enumerate}
\item
\label{point1}
Any bounded-error quantum query subroutine evaluating a Boolean function can be converted into an exact span program whose complexity (called \emph{witness size}) is the same as the query complexity of the subroutine up to a constant factor.%
\footnote{This constant factor is unavoidable here, and it depends on the permitted error probability of the subroutine.  
We revisit this question in \rf{sec:lower-bound}.
}
\item
\label{point2}
Span programs can be composed.  More precisely, a span program for a function $f$ with witness size $W_f$ can be composed with a span program for a function $g$ with witness size $W_g$ resulting in a span program for the composed function $f\circ g$ with witness size $W_fW_g$.
\item
\label{point3}
Any span program can be converted into a quantum query algorithm with the same (up to a constant factor) complexity, introducing some small error.
\end{enumerate}

Thus, having two bounded-error quantum query algorithms evaluating functions $f$ and $g$ in $Q_f$ and $Q_g$ queries, respectively, we can first lift them to the idealised (error-free) world of span programs (\PointRef 1), combine them there (\PointRef 2), and then bring them back (\PointRef 3) resulting in a bounded-error quantum query algorithm for the function $f\circ g$ with query complexity $\OO(Q_fQ_g)$.
This improves on $\OO(Q_f Q_g \log Q_g)$, which we would get by composing the programs for $f$ and $g$ directly using \rf{thm:majorityVoting} or~\ref{thm:ErrorReduction} to reduce the error of the latter.
Of course, this approach also works in more complicated settings, like the recursive procedure mentioned in the first paragraph of this section.

\PointRef 1 above holds, for one reason, because span programs just do not have a notion of error.%
\footnote{
We are talking about the original version of span programs from~\cite{reichardt:formulae} here.
There exists a version of span programs capable of approximate evaluation of functions~\cite{ito2019approxSpan}, but this is just for convenience. We know by \cite{reichardt:spanPrograms} that there is always an \emph{exact} span program with optimal witness complexity.
}
The transformation from the quantum query algorithm to the span program in \PointRef 1 is due to~\cite{reichardt:spanPrograms} and it is indirect.
It goes by proving that span programs are dual to an important quantum lower bound technique: the adversary bound~\cite{hoyer:advNegative}.
Because of the latter, generalisations of span programs beyond evaluation of Boolean functions are generally called \emph{dual adversary bounds}.
\medskip

Span programs allow composition without log factors, but they only work for Boolean functions.
Further research focused on generalising this to other settings like evaluation of arbitrary functions~\cite{lee:stateConversion}, state generation~\cite{ambainis:symmetryAssisted}, and state conversion~\cite{lee:stateConversion}.
In~\cite{belovs:variations}, a version of the dual adversary for state conversion was constructed that allowed the input oracle to be an arbitrary unitary, not just encoding some input string like in~\cite{lee:stateConversion}.
This made it possible to construct idealised counterparts to the whole spectrum of quantum unitary transformations, and 
a \emph{direct} reduction from the quantum query algorithm to the dual adversary was shown (building on~\cite{lee:stateConversion}).

This generality caused some problems with \PointRef 1 above (with span programs replaced by the dual adversary).
If one takes a subroutine for evaluating a function, and converts it into the dual adversary for state conversion, the resulting dual adversary will have precisely the same action, and, consequently, precisely the same error!
If the algorithm had bounded error, the adversary would have the same error.
The remedy proposed in~\cite{belovs:variations} was a \emph{purifier}.
It is a dual adversary that has input oracle of the form~\rf{eqn:InputOracle}, and which performs the transformation~\rf{eqn:ErrorReduction} exactly in constant complexity (assuming $\delta=\Omega(1)$).
Of course, the exactness only holds in the idealised world.
If one were to directly implement the purifier as a real quantum algorithm (like in \PointRef 3), this would bring the bounded error back.
Composing it with the dual adversary obtained from the bounded-error algorithm above, yields a dual adversary for \emph{exact} evaluation of the function with only $\OO(1)$ multiplicative increase in complexity.
The purifier can be used in other settings as well.
For instance, it can be applied to the input of the algorithm, covering the case when the input oracle is noisy.
\medskip

The span programs and the dual adversaries described above only consider \emph{query} complexity of quantum algorithms.
In~\cite{belovs:taming}, the model of \emph{transducers} was introduced.
Transducers are a relaxation of both the usual model of quantum algorithms and the dual adversary in the form of~\cite{belovs:variations}.
We will describe them in detail in \rf{sec:prelim}.
For now, let us note that a transducer is a unitary $S$ in a larger space $\cH\oplus \cL$ that yields a certain unitary $S\DownTransduce_{\cH}$ on the subspace $\cH$, called the 
\emph{transduction action} of $S$.
For $\xi,\tau\in\cH$, we write $\xi\transduce{S}\tau$ and say that $S$ \emph{transduces} $\xi$ into $\tau$ if $S\DownTransduce_{\cH}\xi = \tau$.
The transduction action of the transducer is its idealised action (in the same sense as above), and it can be approximately implemented by executing $S$ several times in a row (see \rf{thm:transducer-algorithm} below).
More precisely, $\OO(W)$ iterations suffice, where $W$ is the so-called \emph{transduction complexity}.
% (roughly corresponding to the witness size of span programs).

As an example, if $S$ is the walk operator, representing one step of an electric quantum walk~\cite{belovs:electicityQuantumWalks} on a graph $G$, and $s$ is the initial vertex of the walk, then we have exact transductions $\ket |s>\transduce{S}\ket|s>$ or $\ket|s> \transduce{S} - \ket|s>$, depending on whether the graph contains marked vertices or not.
For details, see \cite[Section~6]{belovs:taming} or \cite{belovs:phaseHelps}.
Here $S$ is the easy-to-implement unitary, $\cH$ is spanned by $\ket |s>$, and $S\DownTransduce_{\cH}$ is the action we are interested in.

Since transducers are unitary transformations themselves, they naturally capture time and space complexity as well, not only query complexity.%
\footnote{Span programs also capture space rather naturally~\cite{jeffery:spanProgramsSpace}, but they only capture time complexity in the sense that it is possible to convert a quantum algorithm into an \emph{approximate} span program, and then back to an algorithm, in a way that more or less preserves time complexity -- there is no clean way to reason about the time complexity within the span program itself~\cite{cornelissen:spanProgramsTime}.}
The \emph{iteration time} complexity $T(S)$ is just the number of gates necessary to implement $S$ (as a usual unitary), and its \emph{space complexity} is the required number of qubits.
The query complexity $L(S)$ is defined directly and it is based on the notion of quantum Las Vegas query complexity from~\cite{belovs:LasVegas}.
We define the latter formally in \rf{sec:prelimAlgorithms}, but for now it suffices to note that it can depend on the input oracle and the initial state, and it is always smaller than the number of queries made by the algorithm.

Again, we can transform any quantum algorithm into a transducer (corresponding to \PointRef 1 above) preserving Las Vegas query complexity, and we can implement the transduction action of any transducer introducing some error (corresponding to \PointRef 3 above).
It is also possible to compose transducers in various ways (corresponding to \PointRef 2 above).
See \rf{tbl:compare} for an informal comparison between complexity measures of a transducer $S$ and an algorithm $A$ performing the same transformation.

\mytable
\label{tbl:compare}
=====
\negmedskip
\[
\begin{tabular}{rcl}
  Transducer $S$   & &  Algorithm $A$\\
  \hline
  $\xi\transduce{S}\tau$ & & $A \ket{}|\xi> \approx \ket{}|\tau>$\\
   $T(S)W(S)$  & $\leftrightarrow$ & time complexity of $A$\\
   $L(S)$  & $\leftrightarrow$ & query complexity of $A$\\
   space complexity of $S$  & $\leftrightarrow$ & space complexity of $A$
\end{tabular}
\]
-----
Transducers have some notions of complexity that correspond to time, query and space complexity, and we can map back and forth between bounded-error quantum algorithms and transducers in a way that preserves these complexities.
More than that: we can compose transducers to get much better complexity than we would obtain by directly composing the algorithms, which makes them a useful model.
=====

Some remarks related to the exactness condition in \PointRef 1 above are in order.
On one hand, as mentioned above, transducers automatically give exact transduction action for any algorithm based on electric quantum walks.
This includes such important algorithms as Grover's search~\cite{grover:search} and Ambainis' algorithm for element distinctness~\cite{ambainis:distinctness}.
On the other hand, given an arbitrary quantum algorithm, what we can build is a transducer whose transduction action is \emph{precisely} equal to the action of this algorithm.
If the algorithm has bounded error, the transducer inherits it.
This is completely analogous to the situation with the dual adversary mentioned above.
Ref.~\cite{belovs:taming} applied a similar solution as for the dual adversary and constructed an explicit implementation of the purifier from~\cite{belovs:variations} as a transducer.

\subsection{Main Results}

In~\cite{belovs:taming}, a purifier was constructed as a transducer with query and transduction complexity $\OO(1/\delta^2)$.
Note that in contrast to majority voting, this is independent of the target error $\eps$, and therefore already allows composition of bounded-error quantum algorithms without log factor overhead. 
However, the iteration time and space complexity of this purifier were comparable to that of majority voting from~\rf{thm:majorityVoting}.

The main contribution of this paper is a better and simpler construction of a purifier.
It turns out that the most convenient way to cast our results is in a form similar to the one used by the QSP-based error reduction from \rf{thm:ErrorReduction}.
The most natural way to state our main result is to use infinite-dimensional space 
(see \rf{thm:mainInfinite} for the formal statement of the next theorem):

\begin{thm}
[Main Result, Infinite-Dimensional Version]
\label{thm:mainIntro}
There exists a transducer (purifier) that uses infinite-dimensional space and that, given query access to any input oracle $\Oref$ satisfying~\rf{eqn:OrefGeneral}, performs the following transduction exactly:
\begin{equation}
\label{eqn:mainIntro}
\ket |\phi> \transduce{} 
\begin{cases}
\ket |\phi>, &\text{if $p < \tfrac 12$;}\\
- \ket |\phi>, &\text{if $p > \tfrac 12$;}
\end{cases}
\end{equation}
for any normalised $\ket|\phi> \in \spn\sfigA{\ket|0>\ket|\phi_0>, \ket|1>\ket|\phi_1>}$ (including the original vector $\ket|\phi>$ of~\rf{eqn:InputOracle}).
The query complexity of the transducer is $1/(2\delta)$ and the transduction complexity is $\OO(1/\delta)$, where $\delta = |\tfrac12 - p|$.
\end{thm}

Of course, it is not possible to implement this construction directly.
We will explain in a moment how to solve this issue, but for now let us note
that our new purifier is extremely simple: besides executions of the input oracle, it only uses two increment and two decrement operations!
Like the QSP-based error reduction procedure from \rf{thm:ErrorReduction}, it uses the reflecting oracle $\Oref$, and a copy of the state $\ket|\phi>$.
Contrary to \rf{thm:ErrorReduction}, the purifier need not know $\delta$ in advance.
In other words, the same transducer works for all values of $\delta$, and the query complexity is input-dependent.

In \rf{sec:lower-bound}, we prove that this transducer is optimal: every transducer with transduction action~\rf{eqn:mainIntro} must have query complexity at least $1/(2\delta)$ even if $\delta$ is known in advance and there are only two admissible input oracles.
Note that this bound is tight \emph{including} the constant factors! 
The latter are important in this context as they matter for recursions with super-constant depth.
On the other hand, the constant factors in the transduction and time complexities are not so important as they do not multiply during composition.

The following corollary (formal version in \rf{cor:InputOracleErrorReduction}) states the same result in the spirit of \rf{thm:majorityVoting} with the state-generating input oracle $O$ and the initial state $\ket|0>$.

\begin{cor}
\label{cor:IntroAlternativeOracle}
There exists a transducer that uses infinite space and, given query access to any input oracle $O\colon \cA\otimes \cW\to \cA\otimes \cW$ as in~\rf{eqn:InputOracle}, performs the following transduction exactly:
\[
\ket |0> \transduce{}
\begin{cases}
\ket |0>, &\text{if $p < \tfrac 12$;}\\
\ket |1>, &\text{if $p > \tfrac 12$.}
\end{cases}
\]
The query complexity is $1 + 1/(2\delta)$, and the transduction complexity is $\OO(1/\delta)$, where $\delta = \absA|\tfrac12 - p|$.
\end{cor}

The transducer from \rf{thm:mainIntro} is exact, but it uses infinite space.
Luckily, allowing small perturbations, it is not hard to restrict it to finite (and small!) space.
The following result states our main result in the form close to \rf{thm:ErrorReduction}
(see \rf{thm:mainFormal} for the formal statement):

\begin{thm}[Main Result, Informal]
\label{thm:main}
Let $\eps, \delta>0$ be real parameters.
There exists a transducer (purifier) that, given query access to any input oracle $\Oref$ satisfying~\rf{eqn:OrefGeneral}, performs the following transduction
with perturbation $\eps$:
\[
\ket |\phi> \transduce{}
\begin{cases}
\ket |\phi>, &\text{if $p\le \tfrac 12 - \delta$;}\\
-\ket |\phi>, &\text{if $p\ge \tfrac 12 + \delta$;}
\end{cases}
\]
for any normalised $\ket|\phi> \in \spn\sfigA{\ket|0>\ket|\phi_0>, \ket|1>\ket|\phi_1>}$ (including, but not limited to, the specific vector $\ket|\phi>$ from~\rf{eqn:InputOracle}).
The transducer has query complexity at most $1/(2\delta)$ and transduction complexity $\OO(1/\delta)$.
It uses $\OO(\ell)$ additional qubits and its iteration time is $\OO(\ell)$, where $\ell = \log \frac1\delta + \log\log \frac1{\eps}$.
\end{thm}

\mytable
\label{tbl:purifier-comparison}
=====
\negmedskip
    \[
    \!\!\!
    \def\arraystretch{1.25}
    \begin{tabular}{r||c|c|c}
         Complexity & Purifier from~\cite{belovs:variations} & Purifier from \cite{belovs:taming} & Purifier of \rf{thm:main} (new) \\
         \hline
        Space & - &$\OO(s\frac{1}{\delta^2}\log\frac{1}{\eps})$ & $s+\OO(\log\frac{1}{\delta} + \log\log\frac{1}{\eps})$ \\
        Iteration time & - & $\OO(s\frac{1}{\delta^2}\log\frac{1}{\eps})$ & $\OO(\log\frac{1}{\delta} + \log\log\frac{1}{\eps})$ \\
        Query & $\OO(\frac{1}{\delta^2})$ & $\OO(\frac{1}{\delta^2})$ & $\frac{1}{2\delta}$\\
        Transduction & - & $\OO(\frac{1}{\delta^2})$ & $\OO(\frac{1}{\delta})$\\
    \end{tabular}
    \]
-----
A side-by-side comparison of our purifier in \rf{thm:main} with the ones from~\cite{belovs:variations} and~\cite{belovs:taming}. 
The one from~\cite{belovs:variations} is in the form of dual adversary, and only query complexity applies to it.
Ref.~\cite{belovs:variations} does not explicitly evaluate the query complexity, only mentioning that it is $\OO(1)$ for all $\delta = \Omega(1)$, but it is not hard to work out the required estimate.
Iteration time for both transducers in the last two columns is with respect to the circuit model.
=====

Here, the perturbation is a technical term similar to the $\eps$-approximation in \rf{thm:majorityVoting} and~\ref{thm:ErrorReduction}.
Also, iteration time here is with respect to the circuit model, that means that time is defined as the number of 1- and 2-qubit operations.
Indeed, as we mentioned previously, we only have to make two increment and two decrement operations on each iteration.
\rf{tbl:purifier-comparison} provides a side-by-side comparison of the new purifier and the previous ones from~\cite{belovs:variations} and~\cite{belovs:taming}.
The purifiers from~\cite{belovs:variations} and~\cite{belovs:taming} use the state-generating input oracle, but this does not affect the asymptotics in the corresponding columns.
The inverse quadratic dependence on $\delta$ in~\cite{belovs:variations} and~\cite{belovs:taming} was not such a crucial issue, as the purifier could be composed with amplitude estimation like in~\rf{cor:ErrorReduction}, but that would further increase space and iteration time complexity.

There is another subtle difference between the purifiers from~\cite{belovs:taming} and \rf{cor:IntroAlternativeOracle}.
The purifier from~\cite{belovs:taming} only executed the input oracle on the initial state $\ket|0>$.
The purifier from \rf{cor:IntroAlternativeOracle} achieves quadratic improvement by executing the input oracle on initial states different from $\ket|0>$.
This is not unlike the distinction between majority voting of \rf{thm:majorityVoting} and amplitude estimation from \rf{thm:estimation}, the latter also obtaining quadratic improvement by executing the input oracle on non-$\ket|0>$ initial states.
This is not an issue if the input oracle $O$ is implemented by a usual quantum subroutine, whose complexity is the same from every initial state.
On the other hand, if $O$ is a transducer itself, this might cause some complications as it is not sufficient to estimate its complexity only on the initial state $\ket|0>$.

We can also compare the new purifier in \rf{thm:main} with the QSP-based error reduction procedure of \rf{thm:ErrorReduction}.
The main advantage of the purifier is that its query complexity does not depend on $\eps$.
For a fair comparison, it makes sense to assume that the QSP-based error reduction has $\OO\sA[\frac1\delta\log\frac1\eps]$ iterations, each iteration taking $\OO(1)$ time, and the number of iterations of the purifier is given by its transduction complexity (see \rf{thm:transducer-algorithm}).
In this sense, the query complexity and the number of iterations in both approaches has the same inverse linear dependence on $\delta$.
The space and iteration time complexities of the QSP-based error reduction is clearly impossible to beat, but these complexities of the new purifier are so small that they are effectively negligible.

The main drawback of the purifier is that in order to get it working, the program and the subroutine have to be converted into transducers and composed similarly to \PointRef 2 on Page~\pageref{point2}.
The QSP-based algorithm is advantageous if the overhead of this conversion kills the benefits of smaller query complexity.
See also discussion in \rf{sec:discussion}.

Another minor drawback of transducers is that implementation of their transduction action (like in \PointRef 3 on Page~\pageref{point3}, see also \rf{thm:transducer-algorithm}) has bad inversely quadratic dependence on the required precision $\eps$.
Therefore, a natural solution is to implement the transduction action with constant precision, resulting in bounded error, and then either measure the output and use classical majority voting, or, if this is not possible, use the QSP-based error reduction to reduce the error to tolerable value.
The benefit is that this error reduction is needed only once for the whole algorithm, and not for each individual subroutine.

\subsection{Consequences and Open Problems}
\label{sec:discussion}
One set of consequences of our results are similar to the ones obtained in~\cite{belovs:taming}.
We will mention one simple application, \rf{thm:composition} below, which is proven in~\cite{belovs:taming} using the old purifier construction.
The proofs go through in exactly the same way, only the transducer gets replaced, so we do not repeat them here.

\begin{thm}
[{\cite[Theorems 3.15 and 3.16]{belovs:taming}}]
\label{thm:composition}
Let $A$ and $B$ be quantum algorithms in the circuit model that evaluate the
functions $f$ and $g$, respectively, with bounded error.
Let $L$ be the query complexity of $A$, and $T(A)$ and $T(B)$ be the time complexities of $A$ and $B$.
Then, 
\begin{itemize}
\item[(a)]
there exists a quantum algorithm in the circuit model that evaluates the composed function $f\circ g$ with bounded error in time complexity
\begin{equation}
\label{eqn:compositionCircuit}
\OO(L)\sA[T(A) + T(B)].
\end{equation}
\item[(b)]
There exists a quantum algorithm using quantum random access memory that evaluates the composed function $f\circ g$ with bounded error in time complexity
\begin{equation}
\label{eqn:compositionQRAG}
\OO\sA[T(A) + L\cdot T(B)]
\end{equation}
assuming random access and elementary arithmetic operations on memory addresses (word RAM operations) can be done in time $\OO(1)$.
\end{itemize}
\end{thm}

Not only does our new purifier improve the space complexity of the composed quantum algorithm referred to in \rf{thm:composition}, our improved dependence on $\delta$ gives better constant-factor overhead, which could matter a lot in algorithms that are obtained from super-constant-depth composition (using extensions of \rf{thm:composition} for larger recursion depth).
Also, iteration time complexity of our new purifier is so small that it is dominated by other terms in~\rf{eqn:compositionCircuit} and~\rf{eqn:compositionQRAG}.
\medskip

While the expression in~\rf{eqn:compositionQRAG} seems very natural, the one in~\rf{eqn:compositionCircuit} is not so, which is a consequence of the way iteration time behaves under composition of transducers.
\eqrf{eqn:compositionCircuit} is optimal if $T(A) \le T(B)$, though.

\begin{open}
Is it possible to improve~\rf{eqn:compositionCircuit} to $\OO\sA[T(A) + L\cdot T(B)]$ in the circuit model?
\end{open}

Another application of our main result, \rf{thm:mainIntro}, is a proof that $\mathsf{QMA} = \mathsf{QMA}_1$ assuming that the one-sided-error variant of $\mathsf{QMA}$ is allowed to use infinite space~\cite{jeffery:QMA1}.
More generally, it seems quite useful to have an effectively exact representation of bounded-error quantum algorithms, and we suspect this has implications beyond composition. 

Let us mention another open problem, of a more technical nature.
We demonstrated that the query complexity of the transducer with the reflecting oracle in \rf{thm:mainIntro} is the best possible, including the constant factor.
For the state-generating variant \rf{cor:IntroAlternativeOracle}, the picture is not so clear.
Since reflection can be implemented using state generation, the former is more powerful than the latter.
On the other hand, the latter is more relevant from the practical point of view.
It is clear that $\Omega(1/\delta)$ query complexity is necessary, otherwise it would break the lower bound in~\cite{nayak:approximatingMedian} mentioned above, but the precise value is not clear.

\begin{open}
Find out the exact value of the query complexity required to solve the problem of \rf{cor:IntroAlternativeOracle}.
The same question can be also asked if we require the purifier to query the state-generating input oracle only on the initial state $\ket|0>$.
\end{open}

For a high-level didactic overview of some of the concepts of this paper, as well as previous related work, we refer the reader to~\cite{jeffery2024SIGACT}.

\subsection{Intuition}
\label{sec:intuition}
In this section, we give an informal overview of the construction behind \rf{thm:mainIntro} and its connection to the usual classical majority voting (see also \cite{jeffery2024SIGACT}).
We construct our purifier as an electric quantum walk on a line.
Indeed, as mentioned in \rf{sec:spanPrograms}, every electric quantum walk gives rise to a transducer whose transduction action is exact, which we take advantage of here.

Before describing our quantum walk, let us revisit classical majority voting.
Assume we are given a biased coin that outputs $-1$ with probability $1-p$, and $+1$ with probability $p$.
We want to decide whether $p\le \tfrac12 -\delta$ (negative case), or $p\ge \tfrac12 +\delta$ (positive case).
We toss the coin a total of $K$ times.
If the sum of all the outcomes is positive, we accept; otherwise, we reject.
By Hoeffding's inequality, the error probability is exponentially small in $K$ (similar to \rf{thm:majorityVoting}).

\myfigure
\label{fig:line2directions}
=====
\[
\def\nodegap{1.35}
\begin{tikzpicture}
\foreach \x in {-4,...,4}
    \draw (\x*\nodegap, 0) node[draw,shape=circle,inner sep=7] (\x) {} node {$\x$};
\foreach \x [evaluate={\y=int(\x+1)}] in {-4,...,3}
    \draw (\x) -- (\y);
\draw (-4) -- +(-1,0) +(-\nodegap,0) node{$\cdots$};
\draw (4) -- +(1,0) +(\nodegap,0) node{$\cdots$};
\end{tikzpicture}
\]
-----
Majority voting can be seen as a random walk on the infinite line.
=====

This can be seen as a random walk on the infinite line (see \rf{fig:line2directions}), where we start in the vertex $0$, and, on each step, move to the left with probability $1-p$ and to the right with probability $p$.
The current vertex stores the total sum of all the outcomes.
At the end, we accept if we are to the right of the vertex 0.
Using the standard convention that the probability of following an edge is proportional to its weight, this random walk is captured by assigning the weight 
\begin{equation}
\label{eqn:randomWalkWeights}
w_j = \s[\frac{p}{1-p}]^j
\end{equation}
to the edge between the vertices $j$ and $j+1$.
Indeed, the probability of going right with this assignment of weights is precisely $p/(1-p)$ times the probability of going left.

Of course, this random walk still has error, albeit exponentially small.
Informally, we get exactness by running this walk for \emph{infinitely long time}.
Then, if $p<1/2$, we will be at $-\infty$; and if $p>1/2$, we will be at $+\infty$; both with probability 1.

For simplicity of implementation, we cut the line in half, only considering the vertices with non-negative labels.
They form an infinite ray (see \rf{fig:line}).
In this graph, the informal exactness statement, made above for the infinite line, can be made precise using \emph{transience} in the sense of P{\'o}lya~\cite{polya:transience} (see also~\cite[Chapter 2]{doyle:walksElectric}).
If $p<1/2$, which is equivalent to $p/(1-p)<1$, this walk is \emph{recurrent}, meaning that it returns to the vertex 0 infinitely many times with probability 1.
If $p>1/2$, which corresponds to $p/(1-p)>1$, this walk is \emph{transient}, meaning that it will escape to infinity with probability 1.
Our quantum walk can be seen as deciding which is the case.

\myfigure
\label{fig:line}
=====
\[
\def\nodegap{1.7}
\begin{tikzpicture}
\foreach \x in {0,...,5}
    \draw (\x*\nodegap, 0) node[draw,shape=circle,inner sep=7] (\x) {} node {$\x$};
\foreach \x [evaluate={\y=int(\x+1)}, evaluate={\z=int(2*\x)}] in {1,2,3,4}
    \draw (\x) --
        node[below,blue] {$\ket|\x>$}
        node[above] {$\gamma^{\z}$}
            (\y);
\draw (0) --node[below,blue] {$\ket|0>$} node[above] {1} (1);
\draw (5) -- +(1,0) +(\nodegap,0) node{$\cdots$};
\end{tikzpicture}
\]
-----
The quantum walk on the infinite ray used by our purifier.
The number above the edge specifies its weight as in~\rf{eqn:randomWalkWeights}, where we use the same parameter $\gamma = \sqrt{p/(1-p)}$ we will be using in \rf{sec:infiniteDimensional} and beyond.
Below each edge, the corresponding state of the quantum walk is written.
=====

Let us briefly describe the construction of the quantum walk, more detail being postponed till \rf{sec:SimpleInfinite}.
The quantum walk happens on edges.
The state $\ket |j>$ corresponds to the edge between the vertices $j$ and $j+1$.
The quantum walk is composed of \emph{local reflections}.
Each vertex $j>0$ corresponds to the local reflection about the state
\[
\sqrt{w_{j-1}} \ket|j-1> + \sqrt{w_j} \ket|j>,
\]
which is proportional to
\begin{equation}
\label{eqn:walk-states}
\sqrt{1-p}\ket |j-1> + \sqrt{p} \ket|j>.
\end{equation}
This state is similar to the state $\phi$ in~\rf{eqn:InputOracle}, except it does not have the ``garbage'' states $\phi_0$ and $\phi_1$.
For simplicity, we will assume that we can perform a reflection about the state~\rf{eqn:walk-states}.
This is not just to gain intuition: we will make the same assumption in \rf{sec:SimpleInfinite}, where we analyse the same quantum walk with the same simplified oracle formally.
The general case of \rf{thm:mainIntro} turns out to follow quite easily from this special case, see \rf{sec:generalInfinite}.

Returning to quantum walk, one iteration thereof is given by $R_2R_1$, where $R_1$ is the composition of all local reflections about the odd vertices, and $R_2$ about the even vertices.
To upper bound the number of times the algorithm should call $R_2R_1$ -- that is, the \emph{hitting time} -- while ensuring that the algorithm distinguishes between the transient (positive) and recurrent (negative) cases, we need to upper bound the corresponding complexity for each case (for details, see~\cite{belovs:electicityQuantumWalks}).
We will assume here for simplicity that $\delta = \Omega(1)$.
The complexity of the quantum walk in the negative case $p\le \tfrac12 -\delta$ is given by the total weight $W$ of the graph.
It equals
\[
W = \sum_{j=0}^\infty \gamma^{2j} = \OO(1)
\]
since $\gamma = 1 - \Omega(1)$ in this case.

In the positive case $p\ge \tfrac12 +\delta$, the complexity is given by the effective resistance $R$ between the initial and the marked vertices, where resistance of each edge is given by the inverse of its weight.
The marked vertex in this case is infinity, hence, the total resistance is the sum of the resistances of all edges:
\[
R = \sum_e\frac{1}{w_e}=\sum_{j=0}^\infty \frac{1}{\gamma^{2j}} = \OO(1)
\]
since $\gamma = 1 + \Omega(1)$ in this case.

Quite remarkably, the number of iterations of the quantum walk is finite despite the graph being infinite and the corresponding random walk requiring infinitely many steps to reach $+\infty$.
Thus, this gives an instance of a 1-vs-infinity separation between quantum and classical hitting times.
This seems to contradict the popular belief that electric quantum walks cannot give better than quadratic speed-up compared to random walks.
This belief is rooted in the observation that the classical hitting time is bounded by $\OO(WR)$.
However, there is a difference between classical and quantum $W$s and $R$s.
In the classical case, the product $WR$ is computed for each graph, and then maximised over all graphs.
In the quantum case, $W$ is maximised over all negative cases, and $R$ over all positive cases, and only then they are multiplied.
This might cause significant problems in the design of quantum walks~\cite{belovs:learningKdistPrior}.
But in this paper, this distinction works to our advantage.

\subsection{Organisation of the Paper}

Let us explain the organisation of the remainder of this paper.
In \rf{sec:prelim}, we formally define and give the necessary background on transducers. 
In \rf{sec:infiniteDimensional}, we analyse the quantum walk on the infinite ray from \rf{fig:line}.
At first, in \rf{sec:SimpleInfinite}, we assume we have access to an input oracle implementing a reflection about the states in~\rf{eqn:walk-states}.
In \rf{sec:generalInfinite}, we assume the input oracle $\Oref$ from~\rf{eqn:OrefGeneral} with the ``garbage'' states $\phi_0$ and $\phi_1$.
It turns out that the garbage states can be taken care of by performing two quantum walks isomorphic to the simple one from \rf{sec:SimpleInfinite}.

In \rf{sec:finiteDimensions}, we describe how the above purifier can be adopted to work in the finite-dimensional settings.
First, in~\rf{sec:time-efficient}, we consider a time-efficient implementation that introduces some small error.
Next, in \rf{sec:query-efficient}, we describe a variant of the purifier that remains exact, but only works in the settings of the state conversion problem as defined in~\cite{belovs:variations}.
In \rf{sec:lower-bound}, we show that the query complexity of the purifier from \rf{thm:mainIntro} cannot be improved.
In \rf{sec:QSP}, we describe error reduction procedure based on quantum signal processing, proving \rf{thm:ErrorReduction}.
Finally, in \rf{sec:non-Boolean}, we describe how to extend the results to non-Boolean functions using the Bernstein-Vazirani algorithm~\cite{bernstein:quantumComplexity}.

\mycutecommand\Ic{I^\circ}
\mycutecommand\Ib{I^\bullet}
\mycutecommand\Iw{I^\circ}
\mycutecommand\cHw{\cH^\circ}
\mycutecommand\cHq{\cH^\bullet}
\mycutecommand\cHt{\cH^\uparrow}

\section{Preliminaries}
\label{sec:prelim}

In this section, we describe the key notions on quantum Las Vegas query complexity from~\cite{belovs:LasVegas} and  transducers from~\cite{belovs:taming}.

\subsection{Quantum Query Algorithms and Las Vegas Query Complexity}
\label{sec:prelimAlgorithms}

We define quantum query algorithms following the standard form of alternating queries and unitaries~\cite{buhrman:querySurvey}, but we allow the input oracle to be an arbitrary unitary~\cite{belovs:variations, belovs:LasVegas}.

\begin{defn}
\label{defn:queryAlgorithm}
A \emph{quantum query algorithm} is a unitary transformation $A = A(O)$ in some space $\cH$ that has the form
\begin{equation}
\label{eqn:prelimAlgorithm}
A(O) = U_Q\, \tO\, U_{Q-1}\, \tO\, U_{Q-2}\, \tO\, \cdots\, U_2\, \tO\, U_1\, \tO\, U_0.
\end{equation}
Here, $O\colon \cM\to\cM$ is the \emph{input oracle}, which can be an arbitrary unitary; $U_t$ are arbitrary unitaries in $\cH$, and $\tO$ is the \emph{query operator} defined as follows.

The space of the algorithm is decomposed as $\cH = \cHw \oplus\cHq$, where the input oracle acts only on $\cHq$.
For that, we assume $\cHq = \cH^\uparrow\otimes \cM$ for some space $\cH^\uparrow$.
Then, the query operator is
\begin{equation}
\label{eqn:queryOperator}
\tO = \Iw \oplus (I\otimes O),
\end{equation}
where $I^\circ$ and $I$ are the identities on $\cH^\circ$ and $\cH^\uparrow$, respectively.
\end{defn}

We call the resulting unitary $A(O)$ in~\rf{eqn:prelimAlgorithm} the \emph{action} of the algorithm (for the particular input oracle), while we usually reserve the term \emph{algorithm} to the specific product on the right-hand side of the same equation.

Not all query algorithms follow this strict form.
For instance, it is not always the case that all queries are equal to $\tO$.
Sometimes, the controls are different, or the input oracle is applied to different registers.
However, all such algorithms can be converted into the form~\rf{eqn:prelimAlgorithm} using standard techniques.

Note that in \rf{defn:queryAlgorithm} we only consider queries to $O$, but not to its inverse $O^*$, i.e., the algorithm has \emph{unidirectional} access to the input oracle.
This does not matter for oracles that are self-inverse, like the reflecting oracle $\Oref$ from~\rf{eqn:Oref}.
However, for other oracles, like the state-generating oracle $O$ from~\rf{eqn:InputOracle}, we would like to allow \emph{bidirectional} access.
We achieve this by replacing the input oracle $O$ with $O\oplus O^*$.
This allows the algorithm to apply either $O$ or $O^*$ based on the value of some control register.
In this paper, we allow bidirectional access to the state-generating input oracle.

\medskip

\myfigure
\label{fig:queryAlgorithm}
=====
\newcommand{\OneIteration}[1]{
    \edef\indxx{#1}
    \begin{scope}[shift={(4*\indxx-4,0)}]
        \draw (0.9,0.35) rectangle (2.3, 1.65) node[pos=0.5] {\Large $I\otimes O$};
        \draw (3.1,0.15) rectangle (4, 3.35) node[pos=0.5] {\Large $U_{\indxx}$};
        \draw[\witnesscolor,->] (0, 2.65) to node[above]{$\psi_{\indxx}^\circ$} (3.1,2.65);
        \draw[\nonquerycolor,->] (0, 1) to node[above]{$\psi_{\indxx}^\bullet$} (0.9,1);
        \draw[\querycolor,->] (2.3, 1) to (3.1,1);
    \end{scope}
}
\[
\begin{tikzpicture}[every node/.style={font=\scriptsize}, every path/.append style={thick,->},scale=0.95]
\draw[\xicolor] (-1.9,1.9) to node[above]{$\xi$} (-0.9,1.9);
\draw (-0.9,0.15) rectangle (0, 3.35) node[pos=0.5] {\Large $U_{0}$};
\OneIteration{1}
\OneIteration{2}
\OneIteration{3}
\draw[\taucolor] (12,1.9) to node[above]{$\tau$} (13,1.9);
\end{tikzpicture}
\]
-----
A graphical illustration of a quantum query algorithm from~\rf{eqn:prelimAlgorithm} with $Q=3$ queries.
The algorithm interleaves input-independent unitaries $U_t$ with queries $\tO = \Iw \oplus (I\otimes O)$.
The intermediate state $\psi_t$ after $U_{t-1}$ and before the $t$-th query is decomposed as $\psi_t^\circ\oplus \psi_t^\bullet$, where only the second half is processed by the input oracle.
\\
(Everywhere in this paper, arrowed wires in parallel represent systems combined via direct sum, rather than tensor product.)
=====

Now we define several complexity measures related to quantum query algorithms.
The usual definition of the query complexity of the algorithm $A$ is the number of queries it makes, i.e., $Q$ in~\rf{eqn:prelimAlgorithm}.
In particular, it does not depend on the input oracle $O$, nor the initial state $\xi\in\cH$ of the algorithm.

We now define \emph{Las Vegas query complexity}~\cite{belovs:LasVegas}, which introduces dependence on both $O$ and $\xi$.
Assume $A(O)$ is executed on some \emph{initial state} $\xi\in\cH$. 
Let 
\[
\psi_t = U_{t-1}\, \tO\, U_{t-2}\, \cdots\, U_1\, \tO\, U_0\,\xi \in \cH
\]
be the state of the algorithm right before the $t$-th query.
We decompose $\psi_t = \psi_t^\circ \oplus \psi_t^\bullet$ with $\psi_t^\circ\in\cH^\circ$ and $\psi_t^\bullet\in \cH^\bullet$.
Then, the \emph{Las Vegas query complexity} of the algorithm $A$ on the input oracle $O$ and the initial state $\xi$ is defined as 
\begin{equation}
\label{eqn:LasVegas}
L(A,O,\xi) = \sum_{t=1}^Q\normA |\psi_t^\bullet|^2.
\end{equation}
In other words, the Las Vegas query complexity is the sum, over all queries, of the squared norms of the parts of the state to which the query is \emph{actually} applied.
In particular, $L(A,O,\xi)\le Q$ for all normalized $\xi$.

For intuition, note that the Las Vegas (i.e. expected) query complexity of a randomized classical algorithm is the sum over all steps, where a query might be made, of the probability of making a query at that step. Moreover, this notion of quantum query complexity is meaningful in the sense that we can convert $A$ to a (bounded-error) quantum algorithm that makes $\OO(L)$ queries, where $L$ is the maximum of $L(A,O,\xi)$ over all admissible pairs of $O$ and $\xi$.
(This follows from \rf{cor:algorithmToTransducer} and \rf{thm:transducer-algorithm-query} below.)

A more precise description of the work performed by the input oracle is given by the \emph{total query state} defined as
\begin{equation}
\label{eqn:queryState}
q(A, O, \xi) 
= \sum_{t=1}^Q \ket|t>\ket|\psi^\bullet_t> 
= \bigoplus_{t=1}^Q \psi^\bullet_t
\;\in\; \bC^Q\otimes \cH^\bullet
= \bC^Q\otimes \cH^\uparrow \otimes \cM.
\end{equation}
In particular, $L(A,O,\xi) = \normA|q(A,O,\xi)|^2$.
For fixed $A$ and $O$, the mapping $\xi \mapsto q(A,O,\xi)$ is linear.

The \emph{time complexity} of the algorithm $A$ is defined as the total number of elementary operations (gates) required to implement all the unitaries $U_0,U_1,\dots,U_Q$ used by the algorithm.
This complexity depends on the particular model of quantum computation.
In this paper we consider the \emph{circuit model}, where elementary operations are all 1- and 2-qubit operations.
\medskip

\subsection{Perturbations}
\label{sec:perturbed}

We say that an algorithm performs the required transformation \emph{$\eps$-approximately} or it has \emph{imprecision} $\eps$ if the distance between the actual final state of the algorithm and the desired final state is at most $\eps$.
The following folklore lemma is useful in estimating imprecision~\cite{belovs:taming}.

\begin{lem}
\label{lem:surgery}
Assume we have a collection of unitaries $U_1,\dots, U_m$ all acting in the same vector space $\cH$.
Let $\psi_0',\dots,\psi_m'$ be a collection of vectors in $\cH$ such that
\[
\psi_t' = U_t \psi_{t-1}'
\]
for all $t$.
Let $\psi_0,\dots, \psi_m$ be another collection of vectors in $\cH$ such that $\psi_0 = \psi'_0$ and
\[
\normA | \psi_t - U_t \psi_{t-1} | \le \eps_t
\]
for all $t$.
Then,
\[
\normA | \psi_m - \psi'_m | \le \sum_{t=1}^m \eps_t.
\]
\end{lem}

In applications of this lemma, $U_t$ stand for sequential sections of a quantum algorithm.
The vectors $\psi'_t$ form the sequence of states the algorithm goes through during its execution.
The vectors $\psi_t$ form an idealised sequence, which we use instead of $\psi_t'$ in the analysis.
More precisely, the analysis looks like
\begin{equation}
\label{eqn:perturbedAlgorithm}
\psi_0 
\maps{U_1} U_1\psi_1 \stackrel{\eps_1}\approx \psi_2
\maps{U_2} U_2\psi_2 \stackrel{\eps_2}\approx  \psi_3
\maps{U_3} \quad\cdots\quad \stackrel{\eps_{m-1}}\approx  \psi_{m-1}
\maps{U_m} U_m\psi_{m-1} \stackrel{\eps_m}\approx  \psi_m.
\end{equation}
Here $\stackrel{\eps}\approx$ stands for a (conceptual) \emph{perturbation} of size $\eps$, where we are allowed to arbitrarily move the state of the algorithm by a distance at most $\eps$.

\rf{lem:surgery} states that $m$ small perturbations in~\rf{eqn:perturbedAlgorithm} result in one big perturbation
\[
\psi_m \stackrel{\eps}\approx U_m\cdots U_1U_0 \psi_0
\]
of size $\eps = \sum_{t=1}^m \eps_t$.
If $\psi_m$ is the desired final state, the algorithm has imprecision $\eps$.

We can define a \emph{perturbed algorithm} $A = A(O)$ as an algorithm of the form~\rf{eqn:prelimAlgorithm} where in addition to $U_t$ it is allowed to use perturbations.
The total perturbation of the algorithm $A$ is
\begin{equation}
\label{eqn:perturbationOfAlgorithm}
\eta(A, O, \xi) = A(O)\xi - \bar A(O)\xi,
\end{equation}
where $\bar A$ is the same algorithm as $A$ but with all the perturbations removed.

The analysis like the one in~\rf{eqn:perturbedAlgorithm} is ubiquitous in quantum algorithms, where, for example, perturbations are used after executions of bounded-error subroutines (with errors reduced using \rf{thm:majorityVoting}) to replace the actual final state of the subroutine with the final state of the ideal subroutine.
Introducing perturbations inside the algorithm is helpful as it allows better control on the total query state and the Las Vegas query complexity.
In particular, we can avoid subroutines being called on some gibberish initial states whose Las Vegas complexity is hard to estimate.
However, we will not pursue this angle in this paper.

\subsection{Transducers}
\label{sec:transducer}
The transducer is a new paradigm for quantum algorithms where a single unitary --- often simple to implement in some desired sense --- gives rise to some target unitary, which we are interested in.
A transducer is very easy to define.

\begin{defn}
\label{defn:transducer}
A \emph{transducer} is a unitary $S$ acting on a direct sum ${\cal H}\oplus {\cal L}$ of vector spaces, called its \emph{public} and its \emph{private} space, respectively.
\end{defn}

We will usually just write ``a transducer $S$ in $\cH\oplus \cL$'' understanding that the first one is its public, and the second one is its private space.
As one can see from the naming, we have a case of encapsulation here: an outside algorithm only has access to the public space.
The private space is \emph{only} modified by executions of $S$.

\rf{defn:transducer} is somewhat empty as it only specifies the easy-to-implement unitary $S$.
The following theorem yields the corresponding target unitary in $\cH$.

\begin{thm}[{\cite[Theorem~5.1]{belovs:taming}}]
\label{thm:transducer}
    Let $S$ be a transducer in ${\cal H}\oplus {\cal L}$. 
    For any $\xi\in{\cal H}$, there exist unique $\tau = \tau(S,\xi)\in{\cal H}$ and $v=v(S,\xi)\in{\cal L}$ such that
\begin{equation}
\label{eqn:transDefinition}
    S\colon\xi\oplus v\mapsto \tau\oplus v.
\end{equation}
    Moreover, for a fixed $S$, the transformation $\xi\mapsto \tau(S,\xi)$ is unitary.
\end{thm}

\begin{defn}
[Transduction]
For a transducer $S$ in $\cH\oplus \cL$, its \emph{transduction action} is the unitary $S\DownTransduce_{\cH}\colon \cH\to\cH$ defined by the mapping $\xi \mapsto \tau(S,\xi)$ from \rf{thm:transducer}.
We say that $S$ \emph{transduces} $\xi$ into $\tau$ and write $\xi\transduce{S}\tau$ or $S\colon \xi\transduce{}\tau$.
The vector $v(S,\xi)$ is called the \emph{catalyst} of this transduction.
Its norm squared $W(S,\xi)=\norm |{v(S,\xi)}|^2$ is called \emph{transduction complexity}.
As in~\rf{sec:prelimAlgorithms}, we call $\xi$ the \emph{initial state}, and we use the name \emph{initial coupling} for $\xi\oplus v$.
\end{defn}

In general, the catalyst $v$ is not uniquely defined.
This is not an issue since we prove that $\xi\transduce{S}\tau$ using specific $v$ that satisfies~\rf{eqn:transDefinition}, and we call this specific $v$ the catalyst and use it to define transduction complexity (and later query complexity in \rf{defn:transComplexities}).

The following result states that we can approximately implement the transduction action $S\DownTransduce_{\cH}$ using black-box access to $S$.

\begin{thm}[Implementation of Transduction Action {\cite[Theorem~5.5]{belovs:taming}}]
\label{thm:transducer-algorithm}
    Fix a positive integer $K$, and spaces ${\cal H}$ and ${\cal L}$. There is a quantum algorithm that, for any transducer $S$ in ${\cal H}\oplus {\cal L}$
    and initial state $\xi\in {\cal H}$, produces a state $\tau'$ such that
\begin{equation}
\label{eqn:transducer-algorithm}
\normA |{\tau'-\tau(S,\xi)}|\leq 2\sqrt{\frac{W(S,\xi)}{K}}
\end{equation}
    using $K$ controlled calls to $S$, and $\OO(K)$ other elementary operations.
    The algorithm uses $\OO(\log K)$ additional qubits (besides $\cH$ and $\cL$).
\end{thm}

\begin{proof}[Proof sketch]
The space of the algorithm is $(\bC^K\otimes \cH)\oplus \cL$.
The algorithm starts in the initial state $\ket|\xi>$.
As the first step, we attach the uniform superposition over $K$ elements, resulting in the state
\begin{equation}
\label{eqn:transalgorithm1}
\sC[\frac1{\sqrt{K}} \sum_{i=1}^K \ket |i>]\otimes \ket|\xi> = \bigoplus_{i=1}^K \frac{\xi}{\sqrt{K}} \in \bC^K\otimes \cH.
\end{equation}
We perform a conceptual perturbation in the sense of \rf{sec:perturbed} and assume that we are in the state
\begin{equation}
\label{eqn:transalgorithm2}
\sC[\bigoplus_{i=1}^K \frac{\xi}{\sqrt{K}}] \oplus \frac{v}{\sqrt{K}}
\end{equation}
instead.
Now we execute $S$ in total $K$ times, each time to map $\frac{\xi}{\sqrt{K}} \oplus \frac{v}{\sqrt{K}}$ into $\frac{\tau}{\sqrt{K}} \oplus \frac{v}{\sqrt{K}}$.
On each of these executions, we take a fresh copy of $\frac{\xi}{\sqrt{K}}$ from~\rf{eqn:transalgorithm2} and the same $\frac{v}{\sqrt{K}}$ from the register $\cL$ that does not change in the entire process.
At the end, we get the state
\begin{equation}
\label{eqn:transalgorithm3}
\sC[\bigoplus_{i=1}^K \frac{\tau}{\sqrt{K}}] \oplus \frac{v}{\sqrt{K}}.
\end{equation}
We perform another conceptual perturbation and assume that we are in the state
\begin{equation}
\label{eqn:transalgorithm4}
\bigoplus_{i=1}^K \frac{\tau}{\sqrt{K}} = \sC[\frac1{\sqrt{K}} \sum_{i=1}^K \ket |i>]\otimes \ket|\tau> 
\end{equation}
instead.
We end the algorithm by detaching the uniform superposition in the first register.

By \rf{lem:surgery}, the imprecision of the algorithm is bounded by the sum of two perturbations --- between~\rf{eqn:transalgorithm1} and~\rf{eqn:transalgorithm2}, and between~\rf{eqn:transalgorithm3} and~\rf{eqn:transalgorithm4}.
Each of the perturbations is of size $\norm|\frac v{\sqrt K}| = \sqrt{\frac{W(S,\xi)}{K}}$.
\end{proof}

\subsection{Complexities of Transducers}
\label{sec:transducerComplexities}

A transducer $S$ yields a unitary $S\DownTransduce_{\cH}$ on the public space $\cH$ via \SuperRef Theorems\ref{thm:transducer} and~\ref{thm:transducer-algorithm}. 
However, $S$ is also a unitary in $\cH\oplus\cL$ itself, which can be implemented by a quantum algorithm.
We will assume we have some fixed implementation of $S$ by an algorithm, which we will denote by $A_S$ in this section to distinguish from the unitary $S$ per se.
(In other parts of the paper, we will usually identify $S$ with its implementation like we do it for $A$ in \rf{sec:prelimAlgorithms}.)

The transducer $S=S(O)$ can also use an input oracle $O$.
In this case, we may assume that its implementation $A_S=A_S(O)$ is in the form of~\rf{sec:prelimAlgorithms}.
To save brackets, we will write $v(S,O,\xi)$ and $\tau(S,O,\xi)$ instead of $v\sA[S(O),\xi]$ and $\tau\sA[S(O),\xi]$, respectively.

Now we can easily define various complexities of a transducer.

\begin{defn}
[Complexities of a transducer]
\label{defn:transComplexities}
Let $S$ be a transducer.
Its \emph{iteration time}\footnote{
    In~\cite{belovs:taming}, this was called time complexity of $S$, which we now find misleading, as it might suggest that such is the total complexity of the algorithm in~\rf{thm:transducer-algorithm}.
}, $T(S)$, is the time complexity (number of elementary operations) used by its implementation $A_S$.
Its \emph{space complexity} is the space complexity of $A_S$.

For fixed input oracle $O$ and initial state $\xi$, we already defined \emph{transduction complexity} by $W(S,O,\xi) = \|v\|^2$, where $v = v(S,O,\xi)$ is the corresponding catalyst as in~\rf{eqn:transDefinition}.
We also define the \emph{total query state} and the \emph{query complexity} by
\begin{equation}
\label{eqn:transQueryComplexity}
q(S,O,\xi) = q(A_S, O, \xi\oplus v)
\qqand
L(S,O,\xi) = L(A_S, O, \xi\oplus v).
\end{equation}
Note that the state $\xi \oplus v$ above is not normalised, but the definitions in~\rf{eqn:LasVegas} and~\rf{eqn:queryState} still make sense.
\end{defn}

Iteration time is precisely where we marry the model-dependent notion of ``time complexity'' with the nice clean mathematical model of transducers. 
The notion of time complexity we use for $A_S$ might be with respect to a particular gate set; it might include random access gates (or not); it might be the depth. 
The time complexity of the algorithm we get in \rf{thm:transducer-algorithm} will inherit this model. 
Note that neither the iteration time nor the space complexity depend on the initial state or the input oracle.

We have the following simple principle that proves the linearity part of \rf{thm:transducer}:

\begin{clm}[Linearity]
\label{clm:linearity}
If $S$ transduces $\xi_1$ into $\tau_1$ with catalyst $v_1$ and $\xi_2$ into $\tau_2$ with catalyst $v_2$, then $S$ transduces $\alpha \xi_1 + \beta \xi_2$ into $\alpha \tau_1 + \beta \tau_2$ with catalyst $\alpha v_1 + \beta v_2$ for all $\alpha,\beta\in \bC$.
In particular,
\begin{equation}
\label{eqn:linearityQuery}
q\sA[S,O,\alpha \xi_1 + \beta \xi_2] = \alpha q(S,O,\xi_1) + \beta q(S,O,\xi_2).
\end{equation}
Also, for every $\xi\in\cH$ and $\alpha\in \bC$:
\begin{equation}
\label{eqn:linearityScaling}
W(S, O, \alpha \xi) = |\alpha|^2 W(S,O,\xi)
\qqand
L(S, O, \alpha \xi) = |\alpha|^2 L(S,O,\xi).
\end{equation}
\end{clm}

\mycutecommand{\vw}{v^\circ} %v work
\mycutecommand{\vq}{v^\bullet} % v query
\mycutecommand{\Sw}{S^\circ} % S work

\mycutecommand{\cLw}{\cL^\circ}
\mycutecommand{\cLq}{\cL^\bullet}
\mycutecommand\cLt{\cL^\uparrow}

\subsection{Canonical Transducers}
\label{sec:canonical}

We mentioned in the previous section, that we may assume that a transducer $S=S(O)$, as an algorithm in $\cH\oplus \cL$, is in the standard form of~\rf{eqn:prelimAlgorithm}.
However, if we na\"ively applied the algorithm from \rf{thm:transducer-algorithm} to this transducer, this would result in the number of queries $\OO(W\cdot Q)$, where $W$ is the transduction complexity and $Q$ is the number of queries made by the transducer.

This can be improved.
First, it turns out that transducers admit a much simpler canonical form: it is just one(!) execution of the input oracle on a part of the private space, followed by a unitary independent of the input oracle.
Second, there exists a more refined implementation of the transduction action that brings the total number of queries down to the query complexity of the transducer as defined in~\rf{eqn:transQueryComplexity}.
The canonical form of transducers is also important in their composition, which we consider in \rf{sec:composition}.

This section is organised as follows.
First, we define the canonical form.
Then, we state a refined variant of \rf{thm:transducer-algorithm} for canonical transducers.
Finally, we describe how any transducer can be put into canonical form.

\begin{defn}
[Canonical transducer]
\label{defn:canonical}
A transducer $S = S(O)$ with an input oracle $O\colon\cM\to\cM$ is in \emph{canonical form} if it satisfies the following conditions.
First, its private space is decomposed into a direct sum $\cL = \cLw\oplus \cLq$ with $\cLq = \cL^\uparrow \otimes \cM$.
Second, the transducer itself is in the form
\[
    S(O) = S^{\circ}\widetilde{O},
\]
where $\tO = I_{\cH} \oplus \Iw \oplus (I\otimes O)$ is the query operator and $\Sw$ is an arbitrary input-independent \emph{work unitary}.
Here $\Iw$ and $I$ are identities on $\cLw$ and $\cL^\uparrow$, respectively.
\end{defn}

We call $\cLw$ and $\cLq$ the \emph{work} and the \emph{query part} of the private space $\cL$, respectively.
Decomposing the catalyst $v$ appropriately: $v=\vw\oplus\vq$, the action of $S$ looks like (see also \rf{fig:canonical})
\begin{equation}\label{eqn:canonical}
    S(O):\xi\oplus v^{\circ}\oplus v^{\bullet} \maps{\widetilde{O}}
    \xi\oplus v^{\circ}\oplus (I\otimes O) v^{\bullet}
    \maps{S^\circ} \tau\oplus v^{\circ}\oplus v^{\bullet}.
\end{equation}
All the vectors (except $\xi$) in the above equation depend on $S$, $O$, and $\xi$; i.e, we have $\tau = \tau(S,O,\xi)$, and similarly for $v$, $\vw$, and $\vq$.

\myfigure
\label{fig:canonical}
=====
\negbigskip
\[
\begin{tikzpicture}[every path/.append style={thick,->}]
    \draw (2,0) rectangle (3, 3) node[pos=0.5] {$\Sw$};
    \draw (-0.2,0) rectangle (1.2, 1) node[pos=0.5] {$I\otimes O$};
    \draw[\nonquerycolor] (-1,0.5) node[above]{$\vq$} to (-0.2,0.5);
    \draw[\witnesscolor] (-1,1.5) node[above]{$\vw$} to (2,1.5);
    \draw[\xicolor] (-1,2.5) node[above]{$\xi$} to (2,2.5);
    \draw[\nonquerycolor] (3,0.5)  to (4,0.5) node[above]{$\vq$};
    \draw[\witnesscolor] (3,1.5) to (4,1.5) node[above]{$\vw$} ;
    \draw[\taucolor] (3,2.5) to (4,2.5) node[above]{$\tau$} ;
    \draw[\querycolor] (1.2,0.5) to (2,0.5);
\end{tikzpicture}
\]
\negmedskip
-----
A schematic depiction of a transducer in canonical form.
It consists of one application of the input oracle $O$ and an input-independent work unitary $\Sw$.
The catalyst $v\in\cL$ is separated into two parts $v=\vw\oplus \vq$ with $\vw\in\cLw$ and $\vq\in\cLq$. 
The first one is not processed by the oracle, and the second one is.
Note that the input oracle is not applied to the public space.
=====

The total query state and the query complexity take particularly simple form:
\[
q(S, O, \xi) = \vq(S,O,\xi)
\qqand
    L(S,O,\xi) = \normA |v^\bullet(S,O,\xi)|^2.
\]
The transduction complexity is as before
\[
    W(S,O,\xi) = \normA |v(S,O,\xi)|^2 = \normA |\vw(S,O,\xi)|^2 + \normA |\vq(S,O,\xi)|^2.
\]

We have the following refinement of \rf{thm:transducer-algorithm}, which justifies \rf{tbl:purifier-comparison}.
\begin{thm}
[Refined Implementation of Transduction Action {\cite[Theorem 3.3]{belovs:taming}}]
\label{thm:transducer-algorithm-query}
    Fix spaces ${\cal H}$ and ${\cal L}={\cal L}^{\circ}\oplus{\cal L}^\bullet$ with $\cLq = \cL^\uparrow\otimes \cM$, and positive real parameters $W$, $L$, and $\eps$.
 Then there exists a quantum algorithm with query access $O\colon \cM\to\cM$ that, for any canonical transducer $S = S(O)$ in ${\cal H}\oplus {\cal L}$ and any initial state $\xi\in\cH$, transforms the state $\xi$ into $\tau'$ such that
    $$\normA |{\tau'-\tau(S,O,\xi)}|\leq \eps$$
assuming that $W(S,O,\xi)\le W$ and $L(S,O,\xi)\le L$.
The algorithm uses $\OO\sA[L(S,O,\xi)/\eps^2]$ queries to $O$, 
$K=\OO\sA[1+W(S, O,\xi)/\eps^2]$ controlled calls to $S^\circ$, $\OO(K)$ other elementary operations, and $\OO(\log K)$ additional qubits.
\end{thm}

Thus, in particular, the time complexity of the algorithm described in \rf{thm:transducer-algorithm-query} (including executions of $\Sw$, but excluding executions of $O$) is $\OO\sA[1+W(S, O,\xi)/\eps^2]\cdot T(S)$, where $T(S)$, is the iteration time of $S$, which, in agreement to \rf{defn:transComplexities}, is the time complexity of the algorithm implementing $\Sw$.

\medskip

When designing a transducer, treating it as a usual quantum query algorithm from \rf{sec:prelimAlgorithms} is very convenient, and we will continue doing so in this paper.
It is not hard to transform any such transducer into the canonical form with a slight increase in complexity.

\begin{thm}
[Conversion to Canonical Form {\cite[Proposition~10.4]{belovs:taming}}]
\label{thm:canonisation}
Assume $S = S(O)$ is a non-canonical transducer with public space $\cH$ making $Q$ queries to the input oracle $O$.
We also assume that $S$, as a query algorithm, is strictly in the form of~\rf{eqn:prelimAlgorithm}, in particular, all its queries to the input oracle are identical.
Then, there exists a canonical transducer $S' = S'(O)$ with the same public space $\cH$ and the same input oracle that has exactly the same transduction action as $S$.
Moreover, 
\[
q(S',O,\xi) = q(S, O, \xi),\qquad
L(S',O,\xi) = L(S, O, \xi),
\]
and
\[
W(S', O, \xi) = W(S, O, \xi) + L(S, O, \xi).
\]
The transducer $S'$ uses $\OO(\log Q)$ additional qubits and its iteration time is at most $\OO(1)$ times the iteration time of $S$.
\end{thm}

Note that any algorithm $A = A(O)$ as in~\rf{eqn:prelimAlgorithm} can be considered as a transducer $S_A = S_A(O)$ with the same transduction action if we just let the private space $\cL$ to be empty and assume $v=0$ in \rf{eqn:transDefinition}.
This gives the following corollary of the above theorem:

\begin{cor}
\label{cor:algorithmToTransducer}
Assume $A$ is a quantum query algorithm as in~\rf{eqn:prelimAlgorithm}, working in space $\cH$ and making $Q$ queries.
Then, there exists a canonical transducer $S_A$ with the public space $\cH$, whose transduction action is identical to the action of $A$.
Moreover, for any $\xi\in\cH$ and input oracle $O$:
\[
q(S_A,O,\xi) = q(A, O, \xi)
\qqand
W(S_A, O, \xi) = L(S_A,O,\xi) = L(A, O, \xi).
\]
Iteration time complexity of $S_A$ is equal to the time complexity of $A$ up to a constant factor, and $S_A$ uses $\OO(\log Q)$ additional qubits.
\end{cor}

\subsection{Relation to the Adversary Bound}
\label{sec:adversary}
In this section, we describe the relation between transducers and the dual adversary bound for state conversion~\cite{lee:stateConversion, belovs:variations}.
Before we describe the connection, we have to define the latter two objects.

\begin{defn}%[State Conversion Problem]
\label{defn:stateConversion}
A \emph{state conversion problem} is specified by two Hilbert spaces $\cM$ and $\cH$ together with a collection of unitaries $O_x$ in $\cM$ and pairs of states $\xi_x\mapsto\tau_x$ in $\cH$, where $x$ ranges over some set $X$.
We say that an algorithm $A$ (respectively, transducer $S$) \emph{solves} the problem (exactly) if $A(O_x) \xi_x = \tau_x$ (respectively, $\xi_x\transduce{S(O_x)}\tau_x$) for all $x \in X$.
\end{defn}

\begin{defn}
\label{defn:dualAdversary}
The \emph{dual adversary bound} corresponding to the state conversion problem of \rf{defn:stateConversion} is given by the optimisation problem to minimise 
$
\max_{x\in X} \|v_x\|^2,
$
where $v_x\in \cH^\uparrow\otimes \cM$, as $x$ ranges over $X$, are some vectors satisfying the constraint
\begin{equation}
\label{eqn:adversaryOriginal}
\ip<\xi_x, \xi_y> - \ip<\tau_x, \tau_y> = \ipA<v_x, I\otimes \sA[I_\cM - O_x^*O_y]v_y>
\qquad \text{for all $x,y\in X$.}
\end{equation}
Here $\cH^\uparrow$ can be an arbitrary space, and $I$ is the identity on that space.
\end{defn}

Note that~\rf{eqn:adversaryOriginal} can be rewritten in a possibly friendlier equivalent form
\begin{equation}
\label{eqn:adversary}
\ip<\xi_x, \xi_y> - \ip<\tau_x, \tau_y> = \ip<v_x, v_y> - \ipA<(I \otimes O_x)v_x, (I \otimes O_y) v_y>
\qquad \text{for all $x,y\in X$.}
\end{equation}

The optimal value of the dual adversary problem lower bounds the worst-case query complexity of any algorithm solving the corresponding state conversion problem.
This is a corollary of the following result connecting dual adversary to transducers.

\begin{thm}
\label{thm:adversary}
Let $\xi_x$, $\tau_x$, $O_x$ for $x\in X$ define a state conversion problem as in \rf{defn:stateConversion}, and let $v_x$ for $x\in X$ be vectors in $\cH^\uparrow\otimes \cM$.
Consider the following statements:
\begin{itemize}
\item[(a)] There exists a quantum query algorithm $A = A(O)$ such that $A(O_x)\xi_x = \tau_x$ and $v_x = q(A, O_x, \xi_x)$ for all $x\in X$.
\item[(b)] There exists a transducer $S = S(O)$ such that $\xi_x\transduce{S(O_x)}\tau_x$ and $v_x = q(S, O_x, \xi_x)$ for all $x\in X$.
\item[(c)] There exists a canonical transducer satisfying~(b).
\item[(d)] The vectors $v_x$ satisfy the constraints~\rf{eqn:adversary}.
\end{itemize}
Then, the statements~(b), (c), and~(d) are equivalent.
The statement~(a) implies all of them.
\end{thm}

\begin{proof}
The implication (a)$\Longrightarrow$(c) is the content of \rf{cor:algorithmToTransducer}, while the implication (b)$\Longrightarrow$(c) is the content of \rf{thm:canonisation}.
Implication (c)$\Longrightarrow$(b) is trivial as any canonical transducer is a transducer.
It remains to establish equivalence of~(c) and~(d).

Assume there exists a canonical transducer as in \rf{defn:canonical} performing the required transduction.
Then, by~\rf{eqn:canonical}, there exists a unitary $S^\circ$ such that
\[
S^\circ \colon \xi_x \oplus v_x^\circ \oplus (I\otimes O_x) v_x^\bullet
\maps{}
\tau_x \oplus v_x^\circ \oplus v_x^\bullet
\]
for all $x\in X$, where $v_x^\bullet$ are the corresponding total query states.
Existence of such unitary is equivalent to the condition
\[
\ip<\xi_x, \xi_y> + \ip<\vw_x, \vw_y> + \ipA<(I \otimes O_x)\vq_x, (I \otimes O_y) \vq_y>
=
\ip<\tau_x, \tau_y> + \ip<\vw_x, \vw_y> + \ip<\vq_x, \vq_y>
\quad\text{for all $x,y\in X$},
\]
which implies~\rf{eqn:adversary} with $v_x$ replaced by $v^\bullet_x$.

On the other hand, if~\rf{eqn:adversary} is satisfied, then there exists a unitary $\Sw$ such that
\begin{equation}
\label{eqn:query2}
S^\circ \colon \xi_x \oplus (I\otimes O_x) v_x
\maps{}
\tau_x \oplus v_x
\end{equation}
for all $x\in X$.
This is an instance of a canonical transducer with $\cLw$ being empty, $v_x^\circ$ equal to zero, and where the total query states are given by $v_x$.
\end{proof}

This result immediately implies the following corollary.
\begin{cor}
\label{cor:adversary}
The optimal value of the dual adversary problem lower bounds all of the following:
\begin{itemize}
\item[(a)] Worst-case query complexity $\max_{x\in X} L(S,O_x,\xi_x)$ of any transducer $S$ such that $\xi_x\transduce{S(O_x)} \tau_x$ for all $x$.
\item[(b)] Worst-case Las Vegas query complexity $\max_{x\in X} L(A,O_x,\xi_x)$ of any algorithm $A$ such that $A(O_x) \xi_x = \tau_x$ for all $x\in X$.
\item[(c)] The number of queries made by any algorithm from~(b), assuming $\|\xi_x\|=1$ for all $x$.
\end{itemize}
\end{cor}

Lower bounding the optimal value of the dual adversary is the content of the usual (primal) adversary bound.
Note though that \rf{thm:adversary} gives much tighter connection between the dual adversary and the transducers solving the corresponding state conversion problem than just worst-case Las Vegas complexity.

\rf{cor:adversary} only considers algorithms and transducers solving the state conversion problem \emph{exactly}.
In \rf{sec:query-efficient}, we briefly explain how to use purifiers in order to lower bound complexity of \emph{bounded-error} algorithms in the case of function evaluation.

\subsection{Composition of Transducers}
\label{sec:composition}

We can compose transducers similar to how we can compose quantum algorithms, and the complexity of these compositions follows some natural rules. 
In this section, we cover parallel and functional composition. 
Ref.~\cite{belovs:taming} also considers sequential composition.

\paragraph{Algorithms}
To get some intuition, we start with composition of query algorithms as in \rf{sec:prelimAlgorithms}.
We first consider parallel composition, which corresponds to performing algorithms in superposition.
Let $A_i = A_i(O_i)$ with $O_i\colon\cM_i\to\cM_i$ be quantum algorithms for $i=1,\dots,m$.
Their \emph{parallel composition} $\bigoplus_i A_i$ is a quantum algorithm with input oracle $O$ acting in $\bigoplus_i \cM_i$ whose action satisfies
\begin{equation}
\label{eqn:parallel}
(\bigoplus_i A_i) \sB[\bigoplus_i O_i] = \bigoplus_i A_i(O_i)
\end{equation}
for all unitaries $O_i\colon \cM_i\to\cM_i$.
Its total query state is
\begin{equation}
\label{eqn:parallelQueryState}
q\sB[\bigoplus_i A_i, \bigoplus_i O_i, \bigoplus_i \xi_i ]
=
\bigoplus_i q(A_i, O_i, \xi_i).
\end{equation}
If all $A_i$ are equal to some $A$, we can replace $\bigoplus_i A_i$ with $I_m\otimes A$, however, we often drop $I_m$ in this context.
If all $O_i$ are equal to some $O$, we can replace $\bigoplus_i O_i$ with just $O$.
(It is easy to see from~\rf{eqn:queryOperator} that $O$ and $I_m\otimes O$ are equally powerful.)
Finally, the most important case is when all $A_i$ are equal, and all $O_i$ are equal.
In this case, we just write
\begin{equation}
\label{eqn:xiLargerSpace}
\text{$q(A, O, \xi)$ instead of $q(I_\cE\otimes A, O, \xi)$}
\end{equation}
for every $\xi\in\cE\otimes \cH$ in arbitrary $\cE$, since the space $\cE$ can be deduced from $\xi$.
We use the same convention for $L(A,O,\xi)$.
\medskip

Now let us move to functional composition, where one algorithm implements the input oracle of another algorithm.
Assume $A$ is a quantum algorithm as in~\rf{eqn:prelimAlgorithm} with the input oracle $O\colon \cM\to\cM$, 
and let $B$ be an algorithm that acts in $\cM$.
Their \emph{functional composition $A\circ B$} is the algorithm $A$, where each execution of $O$ is replaced by execution of $B$.
The action of $A\circ B$ is $A(B)$, and it does not have the input oracle.

The above description is simple, but it omits a large number of important details.
First, we usually have some global input oracle, which is unrelated to the subroutine $B$.
From now on, we will denote it by $O\colon \cM\to\cM$.
Thus, the algorithm $A$ has access to two input oracles: the global input oracle $O$, and 
$O'\colon \cM'\to \cM'$ that gets replaced by $B$.
The subroutine $B$ now acts in $\cM'$, and it also has access to $O$, as does the functional composition $A\circ B$.

Modelling $A$'s oracle access to two input oracles $O$ and $O'$ in the framework of \rf{sec:prelimAlgorithms} is easy.
We assume it has access to one combined input oracle $O\oplus O'$ acting in $\cM\oplus \cM'$.
We can also measure the query complexity of each of the oracles separately by decomposing the total query state from~\rf{eqn:queryState} as
\[
q(A, O\oplus O', \xi) = q^{(0)} (A, O\oplus O', \xi) \;\oplus\; q^{(1)} (A, O\oplus O', \xi),
\]
where 
$q^{(0)} (A, O\oplus O', \xi) \in \bC^Q\otimes \cH^\uparrow \otimes \cM$
and
$q^{(1)} (A, O\oplus O', \xi) \in \bC^Q\otimes \cH^\uparrow \otimes \cM'$
denote the parts processed by the oracles $O$ and $O'$, respectively.

With these conventions, we can write the action of $A\circ B$ as $A\sA[O\oplus B(O)]$.
It is also not hard to describe the total query state of $A\circ B$ by the following expression~\cite{belovs:LasVegas}:
\begin{equation}
\label{eqn:compositionAlgorithmComplexity}
q\sA[A\circ B, O, \xi] = q^{(0)}\sA[ A, O\oplus B(O), \xi ] \; \oplus \; 
q\sB[B, O, q^{(1)} {\sA[ A, O\oplus B(O), \xi ]}  ].
\end{equation}
Although this expression is lengthy, it is quite natural.
It says that the queries made by $A\circ B$ fall into the queries made by $A$ to $O$ directly  (the first term), and the queries made by the subroutine $B$ (the second term).
For the latter, we use the convention~\rf{eqn:xiLargerSpace}, and we assume that $B$ is executed on the entire total query state in one go, instead of dividing it into $Q$ chunks as $A$ does.
Finally, $B(O)$ replaces $O'$ since $B$ implements this input oracle.

\paragraph{Transducers}
Both parallel and functional composition of transducers mimics the case of algorithms very closely.
Parallel composition $\bigoplus_i S_i$ of transducers $S_i = S_i(O_i)$ can be implemented as parallel composition of their implementations $\bigoplus A_{S_i}$, where we used notation of \rf{sec:transducerComplexities}.
For canonical transducers, it boils down to implementing $\bigoplus \Sw_i$, where $\Sw_i$ are the corresponding work unitaries.
Similarly to~\rf{eqn:parallelQueryState}, we have
\begin{equation}
\label{eqn:parallelTransducerComplexityA}
q\sB[\bigoplus_i S_i, \bigoplus_i O_i, \bigoplus_i \xi_i ]
=
\bigoplus_i q(S_i, O_i, \xi_i)
\end{equation}
and
\begin{equation}
\label{eqn:parallelTransducerComplexityB}
W\sB[\bigoplus_i S_i, \bigoplus_i O_i, \bigoplus_i \xi_i ]
=
\sum_i W(S_i, O_i, \xi_i).
\end{equation}

For every transducer $S = S(O)$ and every space $\cE$, we can define transducer $I_{\cE}\otimes S$, which corresponds to the special case of parallel composition when all $S_i$ are equal to $S$, and all $O_i$ are equal to $O$.
As in~\rf{eqn:xiLargerSpace}, we will write, e.g., $q(S,O,\xi)$ instead of $q(I_{\cE}\otimes A, O,\xi)$ for every $\xi\in\cE\otimes \cH$.
\medskip

Functional composition is slightly more subtle, but similar.
Let $S_A$ and $S_B$ be canonical transducers with public spaces $\cH$, and $\cM'$, respectively.
We assume $S_A$ has input oracle $O\oplus O'$ acting in $\cM\oplus \cM'$, and $S_B$ has input oracle $O$ acting in $\cM$.
Let $U_A(O\oplus O')$ and $U_B(O)$ denote the transduction actions of $S_A$ and $S_B$ on $\cH$ and $\cM'$, respectively.
As for algorithms, we can decompose the total query state of $A$ as
\[
q(A, O\oplus O', \xi) = q^{(0)}(A, O\oplus O', \xi) \oplus q^{(1)}(A, O\oplus O', \xi)
\]
into the parts processed by $O$ and $O'$, respectively.

\begin{thm}
\label{thm:transducerComposition}
In the above assumptions, there exists a canonical transducer $S_A\circ S_B$ on the public space $\cH$ and with the input oracle $O\colon\cM\to\cM$.
Transduction action of $S_A\circ S_B$ equals $U_A\sA[O\oplus U_B(O)]$.
The total query state of $S_A\circ S_B$ satisfies
\begin{equation}
\label{eqn:compositionTransducerComplexity1}
q(S_A\circ S_B, O, \xi) = q^{(0)}\sA[ S_A, O\oplus U_B(O), \xi ] \; \oplus \; 
q\sB[S_B, O, q^{(1)} {\sA[ S_A, O\oplus U_B(O), \xi ]}  ],
\end{equation}
and the transduction complexity is
\begin{equation}
\label{eqn:compositionTransducerComplexity2}
W(S_A\circ S_B, O,\xi) = W\sA[S_A, O\oplus U_B(O), \xi] + W\sB[{S_B, O, q^{(1)}\sA[S_A, O\oplus U_B(O), \xi]}].
\end{equation}
The iteration time of $S_A\circ S_B$ is $T(S_A)+T(S_B)$, and its space complexity is the sum of space complexities of $S_A$ and $S_B$.
\end{thm}

The logic behind formulae~\rf{eqn:compositionTransducerComplexity1} and~\rf{eqn:compositionTransducerComplexity2} is the same as behind~\rf{eqn:compositionAlgorithmComplexity}.

\subsection{Perturbed Transducers}
\label{sec:transPerturbed}

In \rf{sec:transducerComplexities}, we assumed that the algorithm $A_S$ implementing $S$ does its job exactly.
Most of the results in this section hold if we allow small perturbations in the implementation of $S$.

A \emph{perturbed transducer} $S$ is a transducer whose implementation $A_S$ is a perturbed algorithm in the sense of \rf{sec:perturbed}.
The requirement for $S$ to transduce $\xi$ into $\tau$ is still that $S$ maps $\xi\oplus v$ into $\tau\oplus v$ (with the help of the perturbations).
We define the perturbation of the transducer as
\[
\eta(S, O, \xi) = \eta(A_S, O, \xi\oplus v),
\]
where $v$ is the corresponding catalyst, and the right-hand side is as in~\rf{eqn:perturbationOfAlgorithm}.

All the definitions made for usual transducers can be adopted for perturbed transducers, and the implementations in \rf{thm:transducer-algorithm} and~\ref{thm:transducer-algorithm-query} hold for them as well.
For example, the right-hand side of~\rf{eqn:transducer-algorithm} in \rf{thm:transducer-algorithm} estimates the perturbations used in implementation of the transduction action, and \eqrf{eqn:transducer-algorithm} still holds (assuming each execution of $S$ in the theorem is still an execution of the perturbed version of $S$).
The perturbations in $S$ amount to $K$ total perturbations, each equal to $\eta(S,O,\xi)/\sqrt{K}$, which allows to estimate the total perturbation of the resulting algorithm using \rf{lem:surgery}.
It is also possible to estimate the perturbation of the composed transducer like in \rf{sec:composition}, but will not go into detail as it is outside the scope of this paper.
For more detail, see~\cite{belovs:taming}.

\section{Infinite-Dimensional Construction}
\label{sec:infiniteDimensional}

In this section, we describe a version of the purifier in infinite-dimensional space.
In  \rf{sec:SimpleInfinite}, we start with a formal analysis of the quantum walk on the infinite ray from \rf{sec:intuition}, which uses a simplified input oracle.
In \rf{sec:generalInfinite}, we generalize it to the real input oracle $\Oref$ from~\rf{eqn:OrefGeneral}, thus proving \rf{thm:mainIntro}.
In \rf{sec:alternativeInputOracle}, we prove \rf{cor:IntroAlternativeOracle}.

Although the purifiers constructed in this section cannot be implemented directly due to their infinite space requirements, they can still be analysed and they demonstrate all the main ideas in the clearest way.
In \rf{sec:finiteDimensions}, we will modify the purifiers so that they use finite space.

%\mycutecommand{\Osimple}{O_{\mathrm{simple}}}

\subsection{Special Case}
\label{sec:SimpleInfinite}
In this section, we perform a formal analysis of the quantum walk on the infinite ray described in \rf{sec:intuition}.
As there, we assume a simplified input oracle $O_p\colon \bC^2\to\bC^2$ reflecting about the state 
\begin{equation}
\label{eqn:varphi}
\varphi_p = \sqrt{1-p} \ket|0> + \sqrt{p}\ket|1>,
\end{equation}
which is similar to the state in~\rf{eqn:InputOracle}, but without the ``garbage'' parts $\phi_0$ and $\phi_1$.
In this section, we use parameter
\begin{equation}
\label{eqn:gamma}
\gamma = \sqrt{\frac{p}{1-p}}.
\end{equation}
The two extreme cases $p=0$ and $p=1$, corresponding to $\gamma = 0$ and $\gamma=\infty$, respectively, are also covered by the analysis, where we use the usual convention that $0^{-1} = \infty$ and $\infty^{-1}=0$.

The purifier transducer $S = S(O_p)$ acts in one register $\reg N$ which stores a non-negative integer, usually denoted by $j$ in this section.
However, to make the increment and the decrement operations in $\reg N$ unitary, we have to assume that $\reg N$ can contain an arbitrary integer, both positive or negative.
We use $\reg A$ to denote the least significant qubit of $\reg N$, and $\reg K$ for the remaining qubits of $\reg N$.
Thus, $\ket K|i> \ket A |b> = \ket N |2i+b>$ for an integer $i$ and $b\in\{0,1\}$.
The public space of $S$ is spanned by $\ket N |0>$, and the private space by $\ket N |j>$ with $j>0$.
As mentioned above, the states $\ket N|j>$ with $j<0$ are not used.
The transducer
\begin{equation}
\label{eqn:S}
S = R_2R_1,
\end{equation}
is a product of two reflections as defined in \rf{fig:simple}.

\myfigure
\label{fig:simple}
=====
\[
R_1\colon \quad
\begin{quantikz}
\lstick{$\cK$}& \qw & \qw\\
\lstick{$\cA$}& \gate{O_p} &\qw
\end{quantikz}
\qquad
R_2\colon \quad
\begin{quantikz}
\lstick{$\cK$}& \gate[2]{+1} &  \ctrl{1} & \gate[2]{-1} &\qw\\
\lstick{$\cA$}& & \gate{O_p} & & \qw
\end{quantikz}
\]
-----
Two reflections comprising the transducer $S = R_2R_1$.
Both reflections act in the register $\cN = \cA\otimes \cK$, where $\cA$ stores the least significant qubit of $\cN$.
The $+1$ and the $-1$ operators increment and decrement the register $\cN$, respectively.
The control for $O_p$ in $R_2$ is on $\cK$ being non-zero.
=====

As we will shortly see in \rf{clm:simpleReflections}, $R_1$ and $R_2$, restricted to the span of $\ket N |j>$ for $j\geq 0$, are the local reflections around the odd and even vertices in \rf{fig:line}.
This includes the local reflection corresponding to the vertex 0, which reflects about the state $\ket|0>$.

The following vectors will be important in the analysis, where $j$ is a positive integer:
\begin{equation}
\label{eqn:simpleStates}
\begin{split}
u_j &= \ket |j-1> + \gamma \ket|j>\\
&\;\propto \sqrt{1-p}\ket |j-1>+\sqrt{p}\ket |j>
\end{split}
\qqand
\begin{split}
u_j^\perp &= \ket |j-1> - \gamma^{-1} \ket|j>\\
&\;\propto \sqrt{p}\ket |j-1> -\sqrt{1-p}\ket |j>.
\end{split}
\end{equation}
In the extreme case of $\gamma = \infty$, we define $u_j = \ket|j>$ and $u_j^\perp = \ket|j-1>$, which still satisfy the proportionality claims in~\rf{eqn:simpleStates}.
Similarly, if $\gamma=0$, we let $u_j^\perp = - \ket |j>$, which also satisfies the proportionality claim.

\begin{clm}
\label{clm:simpleReflections}
For all odd $j$, we have $R_1 u_j = u_j$ and $R_1 u_j^\perp = -u_j^\perp$.
For all even $j>0$, we have $R_2 u_j = u_j$ and $R_2 u_j^\perp = -u_j^\perp$.
Also, $R_2\ket|0> = \ket |0>$.
\end{clm}

\begin{proof}
By construction, for every non-negative integer $i$, $R_1$ acts as identity on $\ket K|i> \otimes \ket A|\varphi_p> = \sqrt{1-p} \ket|2i> + \sqrt{p}\ket|2i+1>$, and flips the phase of (multiplies by $-1$) its orthogonal complement in the span of $\ket|2i>$ and $\ket |2i+1>$.  The state $u_{2i+1}$ is proportional to the former, and $u_{2i+1}^\perp$ to the latter.
The case of $R_2$ is proven similarly.
\end{proof}

We can now prove that $S=R_2R_1$ is an \emph{exact} transducer with the desired transduction action, and query complexity precisely $1/(2\delta) = |1-2p|^{-1}$.
\begin{thm}
\label{thm:simple}
In the above notation, $S = R_2R_1$ performs the following transduction:
\begin{equation}
\label{eqn:simple}
S(O_p) \colon \ket |0> \transduce{} 
\begin{cases}
\ket |0>, &\text{if $p < \tfrac 12$;}\\
- \ket |0>, &\text{if $p > \tfrac 12$.}
\end{cases}
\end{equation}
The query complexity of the transducer is $L(S, O_p, \ket |0>)=1/(2\delta)$, 
and the transduction complexity is $W(S, O_p,\ket|0>) = \OO(1/\delta)$, where $\delta = |\tfrac12 - p|$.
\end{thm}

\begin{proof}
If $p < \tfrac12$, then $\gamma < 1$, and we define the catalyst $\ket |v> = \sum_{j=1}^\infty\gamma^j\ket |j>$.
The action of the transducer on the initial coupling is given by
\begin{equation}
\label{eqn:SimpleSequence1}
\begin{aligned}
\ket|0>+\ket|v>=&\sum_{j=0}^\infty \gamma^j \ket |j> = \sum_{i=0}^\infty \gamma^{2i} \sA[\ket |2i> + \gamma \ket |2i+1>]\\
\maps{R_1} & 
\sum_{j=0}^\infty \gamma^j \ket |j>
= \ket |0> + \sum_{i=0}^\infty \gamma^{2i+1} \sA[\ket |2i+1> + \gamma \ket |2i+2>]\\
\maps{R_2}&
\sum_{j=0}^\infty \gamma^j \ket |j> = \ket|0>+\ket|v>,
\end{aligned}
\end{equation}
where we used \rf{clm:simpleReflections}.
Therefore, $\ket |0> \transduce{S(O_p)} \ket |0>$.
The input oracle is applied to the states in the brackets above, therefore, by~\rf{eqn:LasVegas}, the Las Vegas query complexity is given by
\begin{multline*}
\sum_{i=0}^\infty \normB|{\gamma^{2i} \sA[\ket |2i> + \gamma \ket |2i+1>] }|^2 + 
\sum_{i=0}^\infty \normB|{\gamma^{2i+1} \sA[\ket |2i+1> + \gamma \ket |2i+2>] }|^2 
\\=
\sum_{j=0}^\infty \gamma^{2j} +  \sum_{j=1}^\infty \gamma^{2j} = \frac{2}{1-\gamma^2} - 1 = \frac1{2\delta}.
\end{multline*}

Now, if $p > \frac 12$, then $\gamma > 1$, and we define $\ket|v>=\sum_{j=1}^\infty (-\gamma)^{-j}\ket|j>$.
The action of the transducer on the initial coupling is given by
\begin{equation}
\label{eqn:SimpleSequence2}
\begin{aligned}
\ket|0>+\ket|v>=&\sum_{j=0}^\infty (-\gamma)^{-j} \ket |j> 
= \sum_{i=0}^\infty \gamma^{-2i} \sA[\ket |2i> - \gamma^{-1} \ket |2i+1>]\\
\maps{R_1} & 
-\sum_{j=0}^\infty (-\gamma)^{-j} \ket |j>
= -\ket |0> + \sum_{i=0}^\infty \gamma^{-2i-1} \sA[\ket |2i+1> - \gamma^{-1} \ket |2i+2>]\\
\maps{R_2}&
-\ket |0> - \sum_{i=0}^\infty \gamma^{-2i-1} \sA[\ket |2i+1> - \gamma^{-1} \ket |2i+2>] \\
&\quad= - \ket|0> + \sum_{j=1}^\infty (-\gamma)^{-j} \ket |j>
=-\ket|0>+\ket|v>.
\end{aligned}
\end{equation}
Therefore, $\ket |0> \transduce{S(O_p)} -\ket |0>$.
The Las Vegas query complexity is
\[
\sum_{j=0}^\infty \gamma^{-2j} +  \sum_{j=1}^\infty \gamma^{-2j} = \frac{2}{1-\gamma^{-2}} - 1 = \frac1{2\delta}.
\]
By a similar computation, we can show that $\normA|v|^2\leq \frac{1}{2\delta}$ in both cases, showing that the transduction complexity is at most $\frac{1}{2\delta}$.
\end{proof}

\begin{rem}
\label{rem:SimpleOrder}
Note that \rf{thm:simple} holds for $S = R_1R_2$ as well.
The sequence~\rf{eqn:SimpleSequence1} is similar with the same catalyst; in the sequence~\rf{eqn:SimpleSequence2}, we should replace the catalyst with $- \sum_{j=1}^\infty (-\gamma)^{-j} \ket|j>$.
\end{rem}

\subsection{General Case}
\label{sec:generalInfinite}

In this section, we prove \rf{thm:mainIntro}.
For that, we define the purifier transducer $S' = S'(\Oref)$ for the general case when the reflecting oracle $\Oref$ in $\cA\otimes \cW$ satisfies~\rf{eqn:OrefGeneral}.
We first describe the transducer, and then explain its relation to the purifier from~\rf{sec:SimpleInfinite}.

The transducer uses the register $\reg J$ for storing an index, which can be an arbitrary non-negative integer.
Again, in order to make the increment and decrement unitary, we assume that $\reg J$ can store an arbitrary integer, but we will not use the negative values.
The transducer also uses the registers $\reg A$ and $\reg W$ acted on by the input oracle.
Note that this time, contrary to \rf{sec:SimpleInfinite}, $\reg A$ is \emph{not} a part of the register $\reg J$, but a separate register.
When describing the states, the registers are listed in this order.
The public space of the transducer is $\ket J|0>\otimes \cA\otimes \cW$, which we identify with $\cA\otimes \cW$.
The purifier transducer is given by 
\begin{equation}
\label{eqn:S'}
S' = R_2'R_1'
\end{equation}
for reflections $R_1'$ and $R_2'$ defined as in \rf{fig:general}. 
Each of $R_1'$ and $R_2'$ makes precisely one query to $\Oref$, thus, $S'$ makes two queries.

The implementation~\rf{eqn:S'} of $S'$ is in the form of \rf{eqn:prelimAlgorithm}, which can be achieved by breaking the controlled-$(-\Oref)$ operation in $R_2'$ into the controlled phase and the controlled-$\Oref$ operations.

\myfigure
\label{fig:general}
=====
\[
R_1'\colon \quad
\begin{quantikz}
\lstick{$\cJ$}& \gate{+1}  & \ctrl{1} &  \gate{-1} & \qw\\
\lstick{$\cA$}& \ctrl[open]{-1} & \gate[2]{\Oref} & \ctrl[open]{-1} & \qw\\
\lstick{$\cW$}& \qw &  & \qw & \qw
\end{quantikz}
\qquad\quad
R_2'\colon \quad
\begin{quantikz}
\lstick{$\cJ$}& \gate{+1}  & \ctrl{1} &  \gate{-1} & \qw\\
\lstick{$\cA$}& \ctrl{-1} & \gate[2]{-\Oref} & \ctrl{-1} & \qw\\
\lstick{$\cW$}& \qw &  & \qw & \qw
\end{quantikz}
\]
-----
Two reflections that give the transducer $S' = R_2'R_1'$.
Closed controls ($\bullet$), including the one from the register $\reg J$, are controlled on the wire being non-zero, whereas open controls ($\circ$) are controlled on the wire being 0. The gates $+1$ and $-1$ denote increment and decrement in $\reg J$, respectively.
=====

The following claim is similar to \rf{clm:simpleReflections}.

\begin{clm}
The operator $R_1'$ acts as identity on the states
\begin{equation}
\label{eqn:States1 NonReflected}
\ket |0>\ket|1>\ket|\phi_1>
\qqand 
\ket |j-1>\ket |0> \ket |\phi_0> + \gamma \ket |j>\ket |1> \ket |\phi_1>
\end{equation}
for all positive integers $j$.
The following states get reflected for all positive integers $j$:
\begin{equation}
\label{eqn:States1 Reflected}
\ket |j-1>\ket |0> \ket |\phi_0> - \gamma^{-1} \ket |j>\ket |1> \ket |\phi_1>.
\end{equation}

The operator $R_2'$ acts as identity on the states
\begin{equation}
\label{eqn:States2 NonReflected}
\ket |0>\ket|0>\ket|\phi_0>
\qqand
\ket |j-1>\ket |1> \ket |\phi_1> - \gamma \ket |j>\ket |0> \ket |\phi_0>
\end{equation}
for all positive integers $j$. 
It reflects the following states for all positive integers $j$:
\begin{equation}
\label{eqn:States2 Reflected}
\ket |j-1>\ket |1> \ket |\phi_1> + \gamma^{-1} \ket |j>\ket |0> \ket |\phi_0>.
\end{equation}
\end{clm}
\begin{proof}
We prove the statements for $R_1'$, as those for $R_2'$ are virtually identical. Since $\ket|0>\ket|1>\ket|\phi_1>$ has $\ket|0>\ket|1>$ in the first two registers, none of the controls gets activated, and so all gates act as the identity. For any $j>0$, using $\gamma=\sqrt{p/(1-p)}$, and the fact that $\Oref$ fixes $\sqrt{1-p}\ket|0>\ket|\phi_0>+\sqrt{p}\ket|1>\ket|\phi_1>$,
\begin{align*}
    \ket|j-1>\ket|0>\ket|\phi_0>+\gamma\ket|j>\ket|1>\ket|\phi_1> &\overset{+1}{\longmapsto} \frac{\ket|j>}{\sqrt{1-p}}\left(\sqrt{1-p}\ket|0>\ket|\phi_0>+\sqrt{p}\ket|1>\ket|\phi_1>\right)\\
    &\overset{\Oref}{\longmapsto} \frac{\ket|j>}{\sqrt{1-p}}\left(\sqrt{1-p}\ket|0>\ket|\phi_0>+\sqrt{p}\ket|1>\ket|\phi_1>\right)\\
    &\overset{-1}{\longmapsto} \ket|j-1>\ket|0>\ket|\phi_0>+\gamma\ket|j>\ket|1>\ket|\phi_1>.
\end{align*}
Similarly, using the fact that $\Oref$ reflects $\sqrt{p}\ket|0>\ket|\phi_0>-\sqrt{1-p}\ket|1>\ket|\phi_1>$, 
\begin{align*}
    \ket|j-1>\ket|0>\ket|\phi_0>-\gamma^{-1}\ket|j>\ket|1>\ket|\phi_1> &\overset{+1}{\longmapsto} \frac{\ket|j>}{\sqrt{p}}\left(\sqrt{p}\ket|0>\ket|\phi_0>-\sqrt{1-p}\ket|1>\ket|\phi_1>\right)\\
    &\overset{\Oref}{\longmapsto} -\frac{\ket|j>}{\sqrt{p}}\left(\sqrt{p}\ket|0>\ket|\phi_0>-\sqrt{1-p}\ket|1>\ket|\phi_1>\right)\\
    &\overset{-1}{\longmapsto} -\left(\ket|j-1>\ket|0>\ket|\phi_0>-\gamma^{-1}\ket|j>\ket|1>\ket|\phi_1>\right).\qedhere
\end{align*}
\end{proof}

\myfigure
\label{fig:walkGeneral}
=====
\negmedskip
\[
\def\nodegap{2.2}
\begin{tikzpicture}
\foreach \i in {0,1}
    \foreach \x in {0,...,5}
        \draw (\x*\nodegap, -\i*1.5) node[draw,shape=circle,inner sep=5] (\i\x) {} node {$\x$};
\draw (00) --node[below,blue] {$\ket|0>\ket|0>\ket|\phi_0>$} (01);
\draw (01) --node[below,blue] {$\ket|1>\ket|1>\ket|\phi_1>$} (02);
\draw (02) --node[below,blue] {$-\ket|2>\ket|0>\ket|\phi_0>$} (03);
\draw (03) --node[below,blue] {$-\ket|3>\ket|1>\ket|\phi_1>$} (04);
\draw (04) --node[below,blue] {$\ket|4>\ket|0>\ket|\phi_0>$} (05);
\draw (05) --node[below,blue] {$\ket|5>\ket|1>\ket|\phi_1>$} ++(\nodegap,0) node (06) {} +(0.5,0) node{$\cdots$};
\draw (10) --node[below,blue] {$\ket|0>\ket|1>\ket|\phi_1>$} (11);
\draw (11) --node[below,blue] {$-\ket|1>\ket|0>\ket|\phi_0>$} (12);
\draw (12) --node[below,blue] {$-\ket|2>\ket|1>\ket|\phi_1>$} (13);
\draw (13) --node[below,blue] {$\ket|3>\ket|0>\ket|\phi_0>$} (14);
\draw (14) --node[below,blue] {$\ket|4>\ket|1>\ket|\phi_1>$} (15);
\draw (15) --node[below,blue] {$-\ket|5>\ket|0>\ket|\phi_0>$} ++(\nodegap,0) node (16) {}+(0.5,0) node{$\cdots$};
\foreach \i in {0,1}
    \foreach \x [evaluate={\y=int(\x+1)}] in {0,...,5}
        \draw (\i\x) --node[red, above]{$\ket N_\i|\x>$} (\i\y);
\end{tikzpicture}
\]
-----
As shown in \CrossRef{Lemmata}{lem:one} and~\ref{lem:two}, the purifier in this section can be seen as a quantum walk on two copies of the ray graph from \rf{fig:line}. Just as in \rf{fig:line}, the weight of the edge between vertices $j$ and $j+1$, in either copy, is $\gamma^{2j}$ (not pictured).
The top ray represents \rf{lem:one}, and the bottom one~\rf{lem:two}.
Below each edge, the state that encodes it in the quantum walk is written.
The pattern repeats with period 4.
All these states are pairwise orthogonal.
Above the edge, the corresponding basis state of the register $\cN_0$ or $\cN_1$ from the proof of the corresponding lemma is given.
=====

The purifier $S'=R_2'R_1'$ can be interpreted as a quantum walk on a graph consisting of two disjoint rays, as described in \rf{fig:walkGeneral}. Concretely, in \rf{lem:one}, we show $\ket|0>\ket|0>\ket|\phi_0>\transduce{S'(\Oref)}(-1)^r\ket|0>\ket|0>\ket|\phi_0>$, where $r=0$ if $p<1/2$ and $r=1$ if $p>1/2$; and a similar statement for $\ket|0>\ket|0>\ket|\phi_1>$ in \rf{lem:two}; and finally, in \rf{thm:mainInfinite}, we combine these to show that $\ket|0>\ket|\phi>\transduce{S'(\Oref)}(-1)^r\ket|0>\ket|\phi>$, proving the main theorem of this section.

\begin{lem}
\label{lem:one}
In the above notation, the transducer $S'$, given an input oracle $\Oref$ satisfying~\rf{eqn:OrefGeneral}, has the following transduction action:
\begin{equation}
\label{eqn:1 transduction}
S'(\Oref)\colon \ket |0>\ket|0>\ket|\phi_0> \transduce{} 
\begin{cases}
\ket |0>\ket|0>\ket|\phi_0>, &\text{if $p < \tfrac 12$;}\\
- \ket |0>\ket|0>\ket|\phi_0>, &\text{if $p > \tfrac 12$.}
\end{cases}
\end{equation}
The query complexity and the transduction complexity of $S'$ are $L(S', \Oref, \ket |0>\ket |0> \ket |\phi_0>)=1/(2\delta)$ and $W(S', \Oref,\ket |0>\ket |0> \ket |\phi_0>) =  \OO(1/\delta)$, where $\delta = |\frac12 - p|$.
\end{lem}
\begin{proof}
We will construct an invariant subspace $\cN_0$ of the operator $S'(\Oref)$ that contains the vector $\ket|0>\ket|0>\ket|\phi_0>$.
We will define an isometry from $\cN$ as defined in \rf{sec:SimpleInfinite} onto $\cN_0$ such that $\ket N|0>$ gets mapped into $\ket|0>\ket|0>\ket|\phi_0>$, and the action of $S'(\Oref)$ in $\cN_0$ is identical (under this isometry) to the action of $S(O_p)$ in $\cN$.
The result then follows from \rf{thm:simple}.

To describe the isometry we will define an orthonormal system of vectors $\ket \cN_0|j>$ such that the isometry maps $\ket N|j>$ into $\ket N_0|j>$ for all non-negative integers $j$.
These vectors are defined by
\begin{align*}
\ket N_0|4i> &= \ket|4i>\ket|0>\ket|\phi_0>,\\
\ket N_0|4i+1> &= \ket|4i+1>\ket|1>\ket|\phi_1>,\\
\ket N_0|4i+2> &= -\ket|4i+2>\ket|0>\ket|\phi_0>,\\
\ket N_0|4i+3> &= -\ket|4i+3>\ket|1>\ket|\phi_1>,
\end{align*}
as $i$ ranges over non-negative integers (see \rf{fig:walkGeneral}, the top ray).
It is easy to check, using~\rf{eqn:States1 NonReflected} and~\rf{eqn:States1 Reflected}, that, for all positive odd values of $j$, $R_1'$ acts as identity on the states
\begin{equation}
\label{eqn:1 NonReflected}
\ket N_0 |j-1> + \gamma \ket N_0 |j>
\end{equation}
and reflects the states
\begin{equation}
\label{eqn:1 Reflected}
\ket N_0 |j-1> - \gamma^{-1} \ket N_0|j>.
\end{equation}
Similarly, using~\rf{eqn:States2 NonReflected} and~\rf{eqn:States2 Reflected}, we get that $R_2'$ acts as identity on $\ket N_0|0>$,
and, for all positive even values of $j$, $R_2'$ acts as identity on the states in~\rf{eqn:1 NonReflected} and reflects the states in~\rf{eqn:1 Reflected}.

Comparing this with \rf{clm:simpleReflections}, we get that the action of $R_2'R_1'$ in $\cN_0$ is identical to that of $R_2R_1$ in $\cN$, which proves the lemma.
\end{proof}

\begin{lem}
\label{lem:two}
In the above notation, the transducer $S'$, given an input oracle $\Oref$ satisfying~\rf{eqn:OrefGeneral}, has the following transduction action: 
\begin{equation}
\label{eqn:2 transduction}
S'(\Oref)\colon \ket |0>\ket|1>\ket|\phi_1> \transduce{} 
\begin{cases}
 \ket |0>\ket|1>\ket|\phi_1>, &\text{if $p < \tfrac 12$;}\\
-\ket |0>\ket|1>\ket|\phi_1>, &\text{if $p > \tfrac 12$.}
\end{cases}
\end{equation}
The query complexity and the transduction complexity of $S'$ are $L(S', \Oref, \ket |0>\ket |1> \ket |\phi_1>)=1/(2\delta)$ and $W(S', \Oref, \ket |0>\ket |1> \ket |\phi_1>) =  \OO(1/\delta)$, where $\delta = |\frac12 - p|$.
\end{lem}

\begin{proof}
The idea is similar to \rf{lem:one}.
This time we define an invariant subspace $\reg N_1$ and an isometry from $\cN$ onto $\cN_1$ that maps $\ket N|j>$ into $\ket N_1|j>$ for all non-negative integers $j$.
The vectors are defined by
\begin{align*}
\ket N_1|4i> &= \ket|4i>\ket|1>\ket|\phi_1>,\\
\ket N_1|4i+1> &= -\ket|4i+1>\ket|0>\ket|\phi_0>,\\
\ket N_1|4i+2> &= -\ket|4i+2>\ket|1>\ket|\phi_1>,\\
\ket N_1|4i+3> &= \ket|4i+3>\ket|0>\ket|\phi_0>,
\end{align*}
as $i$ ranges over non-negative integers (see \rf{fig:walkGeneral}, the bottom ray).
Using~\rf{eqn:States2 NonReflected} and~\rf{eqn:States2 Reflected}, 
it is not hard to check that, for all positive odd values of $j$, $R_2'$ acts as identity on the states
\begin{equation}
\label{eqn:2 NonReflected}
\ket N_1 |j-1> + \gamma \ket N_1 |j>
\end{equation}
and reflects the states
\begin{equation}
\label{eqn:2 Reflected}
\ket N_1 |j-1> - \gamma^{-1} \ket N_1|j>.
\end{equation}
Similarly, we get that $R_1'$ acts as identity on $\ket N_1|0>$ and, for all positive even $j$, $R_1'$ acts as identity on the states in~\rf{eqn:2 NonReflected} and reflects the states in~\rf{eqn:2 Reflected}.

Comparing this with \rf{clm:simpleReflections}, we get that the action of $R_2'R_1'$ in $\cN_1$ is identical to that of $R_1R_2$ in $\cN$.
The lemma follows by \rf{rem:SimpleOrder}.
\end{proof}

\begin{thm}
\label{thm:mainInfinite}
The transducer $S'$ defined in~\rf{eqn:S'}, given an input oracle $\Oref$ satisfying~\rf{eqn:OrefGeneral}, has the following transduction action:
\begin{equation}
\label{eqn:mainInfinite}
S'(\Oref) \colon \ket |\phi> \transduce{} 
\begin{cases}
\ket |\phi>, &\text{if $p < \tfrac 12$;}\\
- \ket |\phi>, &\text{if $p > \tfrac 12$;}
\end{cases}
\end{equation}
for any normalised $\ket|\phi> \in \spn\sfigA{\ket|0>\ket|\phi_0>, \ket|1>\ket|\phi_1>}$ (including the original vector $\ket|\phi>$ of~\rf{eqn:InputOracle}).
The query complexity of the transducer is $L(S',\Oref,\phi)=1/(2\delta)$, and the transduction complexity is $W(S', \Oref, \phi) = \OO(1/\delta)$, where $\delta = |\tfrac12 - p|$.

The transducer uses infinite space.
It can be put into canonical form (still using infinite space) so that its query complexity stays the same, and its transduction complexity remains $\OO(1/\delta)$. 
\end{thm}

\begin{proof}
Recall that we identify $\ket|0>\otimes \cA\otimes \cW$, which is the public space of the transducer $S'$, with $\cA\otimes \cW$.
Let $r = 0$ if $p<1/2$; and $r=1$ if $p>1/2$.
By \rf{lem:one}, we have 
\begin{equation}
\label{eqn:mainInfinite1}
\ket|0>\ket|\phi_0>\transduce{S'(\Oref) } (-1)^r \ket|0>\ket|\phi_0>.
\end{equation}
Let $q_0$ be the corresponding total query state $q_0 = q(S', \Oref,\ket|0>\ket|\phi_0>)$.
By the proof, we have that $S'$ acts in $\cN_0$, hence, $q_0\in \bC^2\otimes \cN_0$ as $S'$ makes two queries.
Also, $\|q_0\|^2 = 1/(2\delta)$.
Similarly, by \rf{lem:two}, 
\begin{equation}
\label{eqn:mainInfinite2}
\ket|1>\ket|\phi_1>\transduce{S'(\Oref) } (-1)^r \ket|1>\ket|\phi_1>.
\end{equation}
Let $q_1 = q(S', \Oref, \ket|1>\ket|\phi_1>)$.
Similarly as above, $q_1 \in \bC^2\otimes \cN_1$ and $\|q_1\|^2 = 1/(2\delta)$.

By linearity \rf{clm:linearity}, for all $\alpha, \beta\in \bC$, we get from~\rf{eqn:mainInfinite1} and~\rf{eqn:mainInfinite2}:
\begin{equation}
\label{eqn:mainInfiniteCombination}
\alpha \ket|0> \ket|\phi_0> + \beta \ket|1> \ket|\phi_1> 
\transduce{S'(\Oref)} 
(-1)^r\sA[\alpha \ket|0> \ket|\phi_0> + \beta \ket|1> \ket|\phi_1>] ,
\end{equation}
and the corresponding total query state is $\alpha q_0 + \beta q_1$.
Since $q_0$ and $q_1$ belong to orthogonal subspaces, we have that the query complexity of this transduction is
\[
\normA |{\alpha q_0 + \beta q_1}|^2 = |\alpha|^2 \|q_0\|^2 + |\beta|^2 \|q_1\|^2 = 1/(2\delta)
\]
assuming $|\alpha|^2 + |\beta|^2 = 1$.
The bound on the transduction complexity is proven similarly.
The final statement about conversion into the canonical form follows directly from \rf{thm:canonisation}.
\end{proof}

\subsection{Purifier for the State-Generating Input Oracle}
\label{sec:alternativeInputOracle}

In \rf{thm:mainInfinite}, we established that the purifier $S'$ exactly encodes the value of the Boolean function into the phase, given one copy of the state $\ket|\phi>$.
It is good enough for applications.
For completeness though, we obtain a purifier in the usual form similar to the one stated in  Eqs.~\rf{eqn:InputOracle} and~\rf{eqn:ErrorReduction}.

\mycutecommand{\purifier}{S_{\mathrm{pur}}}
\mycutecommand{\Sref}{S_{\mathrm{ref}}}
\begin{cor}\label{cor:InputOracleErrorReduction}
There exists a transducer $\purifier$ with bidirectional access to an input oracle $O\colon \cA\otimes \cW\to \cA\otimes \cW$ as in~\rf{eqn:InputOracle}, which performs the following transduction exactly:
\begin{equation}
\label{eqn:Spur}
\purifier(O\oplus O^*)\colon \ket |0> \transduce{}
\begin{cases}
\ket |0>, &\text{if $p < \tfrac 12$;}\\
\ket |1>, &\text{if $p > \tfrac 12$.}
\end{cases}
\end{equation}
The query complexity $L\sA[\purifier, O\oplus O^*, \ket|0>] = 1 + 1/(2\delta)$, where $\delta = \absA|\tfrac12 - p|$, and the transduction complexity is $W\sA[\purifier, O\oplus O^*, \ket|0>] = \OO(1/\delta)$.
The transducer uses infinite space.
\end{cor}

\begin{proof}
In the proof, we again denote $r=0$ if $p<1/2$, and $r=1$ if $p>1/2$.

Let us first assume temporarily that we have algorithms (not transducers) performing the required actions, and we will explain how to combine them.
The transducer result then follows by using compositions of transducers instead of algorithms.

Let $U' = S'(\Oref)\DownTransduce_{\cH}$, which is given by~\rf{eqn:mainInfinite}.
Assume we have an algorithm (not transducer) performing $U'$.
Then, it is straightforward to obtain the transformation as in~\rf{eqn:Spur} as follows.
(In the transformation below, the first register is the qubit used in~\rf{eqn:Spur}, and the second register is $\cA\otimes\cW$ as in~\rf{eqn:InputOracle}.)
\begin{equation}
\label{eqn:longAlgorithm}
\begin{aligned}
\ket |0> \ket |0> & \longmapsto \frac1{\sqrt 2} \sA[\ket|0>\ket|0> + \ket|1>\ket|0>] 
    && \text{Hadamard on the first register} \\
& \longmapsto \frac1{\sqrt 2} \sA[\ket|0>\ket|0> + \ket|1>\ket|\phi>] 
    && \text{controlled application of $O$} \\
& \longmapsto \frac1{\sqrt 2} \sA[\ket|0>\ket|0> + (-1)^r \ket|1>\ket|\phi>]\quad 
    && \text{controlled application of $U'$} \\
& \longmapsto \frac1{\sqrt 2} \sA[\ket|0>\ket|0> + (-1)^r \ket|1>\ket|0>] 
    && \text{controlled application of $O^*$} \\
& \longmapsto \ket|r>\ket|0> 
    && \text{Hadamard on the first register} \\
\end{aligned}
\end{equation}
One problem, though, is that $S'$ uses the input oracle $\Oref$, not $O$.
But it is simple to simulate the former with the latter, as by~\rf{eqn:Oref} we have $\Oref =  O \mathrm{Ref_{\ket|0>\ket|0>}} O^*$.

Let $A = A(O\oplus O^*\oplus O')$ be the algorithm in~\rf{eqn:longAlgorithm}, where $O'$ is the placeholder for $U'$.
We have
\begin{equation}
\label{eqn:IOER1}
L^{(0)}(A, O\oplus O^*\oplus U', \ket|0>\ket|0>) = \frac12+\frac12 = 1,
\end{equation}
where the superscript $(0)$ corresponds to the input oracle $O\oplus O^*$.
This is because $O\oplus O^*$ is called twice, each time on a vector of norm $1/\sqrt{2}$. 
Also, we have
\begin{equation}
\label{eqn:IOER2}
q^{(1)}(A, O\oplus O^*\oplus U', \ket|0>\ket|0>) = \frac{(-1)^r\phi}{\sqrt{2}},
\end{equation}
where the superscript $(1)$ corresponds to the input oracle $O'$.
By \rf{cor:algorithmToTransducer}, there exists a canonical transducer $S_A$ whose transduction action is identical to the action of $A$, and that has the same total query state as $A$.
Its transduction complexity is $\OO(1)$.

Similarly, let $A_{\text{ref}} = A_{\text{ref}}(O\oplus O^*)$ be the algorithm from~\rf{eqn:Oref} whose action is equal to $\Oref$.
It satisfies
\begin{equation}
\label{eqn:IOER3}
L(A_{\text{ref}}, O\oplus O^*, \xi) = 2
\end{equation}
for every normalised $\xi\in\cE\otimes \cA\otimes \cW$, where we use the convention from~\rf{eqn:xiLargerSpace}.
Again, by \rf{cor:algorithmToTransducer} there exists a canonical transducer $S_{\text{ref}}$ whose transduction action is equal to the action of $A_{\text{ref}}$ and that has query and transduction complexity equal to 2 on every normalised initial state.

We get the required transducer $\purifier$ as $S_A\circ S'\circ S_{\text{ref}}$.
Using \rf{thm:transducerComposition} to evaluate query complexity of composed transducers, we have
\begin{align*}
L(\purifier, O\oplus O^*, \ket|0>\ket|0>)
&=
L^{(0)}\sA[S_A, O\oplus O^*\oplus U', \ket|0>\ket|0>] \\
&\qquad+ L\sA[S'\circ S_{\text{ref}}, O\oplus O^*,\tfrac{(-1)^r\phi}{\sqrt{2}}] &&\text{by~\rf{eqn:IOER2}} \hspace{-5cm}
\\
&=1 + \frac12 L(S'\circ S_{\text{ref}}, O\oplus O^*, \phi) 
&& \text{by~\rf{eqn:IOER1} and~\rf{eqn:linearityScaling}}
\\
&= 1 + \frac12 L\sA[S_{\text{ref}}, O\oplus O^*, q(S', \Oref, \phi)] 
&& \text{since $S'$ only queries $\Oref$}
\\
&= 1 + \frac12 \cdot 2 \cdot \|q(S', \Oref, \phi)\|^2
&&\text{by~\rf{eqn:IOER3} and~\rf{eqn:linearityScaling}}
\\
&= 1 + \frac12 \cdot 2 \cdot \frac1{2\delta} = 1 + \frac1{2\delta}
&&\text{by \rf{thm:mainInfinite}}. 
\end{align*}
In a similar way, we get that the transduction complexity of $\purifier$ is $\OO(1/\delta)$.
\end{proof}

\section{Finite-Dimensional Implementations}
\label{sec:finiteDimensions}
The purifiers of \CrossRef{Theorems}{thm:simple} and~\ref{thm:mainInfinite} cannot be implemented directly because they use registers ($\cN$ and $\cJ$, respectively) of infinite dimension.
In this section, we consider two different variants of restricting the purifier to finite-dimensional space.

The first one is quite obvious: we replace the infinite-dimensional register with a $D$-dimensional register storing values between $0$ and $D-1$.
This introduces small perturbation, but the overall structure of the algorithm is preserved.
In particular, it can be implemented time-efficiently using the circuits from \CrossRef{Figures}{fig:simple} and~\ref{fig:general}.

The second one is restricted to the special case of the state-conversion problem defined in \rf{sec:adversary}, where we have to implement the transformation $\xi_x\mapsto \tau_x$ for a finite number of pairs.
In this case, we just restrict the infinite-dimensional space of the transducer to the finite-dimensional space spanned by the vectors used by these transformations.

\subsection{Time-Efficient Implementation}
\label{sec:time-efficient}

As mentioned above, to obtain a time-efficient implementation, we replace the infinite-dimensional register with a $D$-dimensional register $\bC^D$ with the basis $\ket|0>,\dots,\ket|D-1>$.
We will call it the truncated version (at depth $D$) of the purifier.
We assume $D\ge 4$ is a power of 2, which makes implementation easier.

\mycutecommand{\trunc}{^{[D]}}
\mycutecommand{\ptrunc}{^{\prime[D]}}

The truncated version $S\trunc = S\trunc(\Oref)$ of the purifier $S$ from \rf{sec:SimpleInfinite} is similarly to~\rf{eqn:S} defined by
\[
S\trunc = R_2\trunc R_1\trunc,
\]
where $R_1\trunc$ and $R_2\trunc$ are defined as $R_1$ and $R_2$ in~\rf{fig:simple}, but where the register $\reg N$ stores an integer between $0$ and $D-1$ and the operations are performed modulo $D$. 

The truncated version $S\ptrunc = S\ptrunc(\Oref)$ of the purifier $S'$ from \rf{sec:generalInfinite} is defined similarly to~\rf{eqn:S'} by
\[
S\ptrunc = R_2\ptrunc R_1\ptrunc,
\]
where $R_1\trunc$ and $R_2\trunc$ are defined as $R_1'$ and $R_2'$ in~\rf{fig:general} with the same replacement for the register $\reg J$.
The following claim is obvious.
\begin{clm}
Both truncated purifiers $S\trunc$ and $S\ptrunc$ use at most $\log D$ qubits besides the space used by the input oracle $\Oref$.
Iteration times $T\sA[S\trunc]$ and $T\sA[S\ptrunc]$ are $\OO(\log D)$.
\end{clm}

From now on, we will concentrate on $S\trunc$, as the arguments for $S\ptrunc$ are similar.
For the former, we have the following simple claim.
\begin{clm}
\label{clm:sameAction}
The actions of $R_1\trunc$ and $R_1$ are the same on all vectors in the span of $\{\ket|0>,\dots,\ket|D-1>\}$.
The actions of $R_2\trunc$ and $R_2$ are the same for all vectors in 
the span of $\{\ket|0>,\dots,\ket|D-2>\}$.
Also, $R_2\trunc\ket|D-1> = \ket|D-1>$ (contrary to $R_2$).
\end{clm}

\begin{proof}
Indeed, repeating the argument from the proof of \rf{clm:simpleReflections}, $R_1\trunc$ and $R_1$ act in the same way on all $u_j$ and $u_j^\perp$ for $j=1,3,5,\dots, D-1$, which are defined in~\rf{eqn:simpleStates}.
These vectors span the subspace $\spn\{\ket|0>,\dots,\ket|D-1>\}$, which proves the claim for $R_1\trunc$.

Using similar reasoning, $R_2\trunc$ and $R_2$ act in the same way on $\ket |0>$ and all $u_j$ and $u_j^\perp$ for $j=2,4,6,\dots, D-2$.
These vectors span the subspace $\spn\{\ket|0>,\dots,\ket|D-2>\}$.
The equality $R_2\trunc\ket|D-1> = \ket|D-1>$ is proven directly. 
\end{proof}

There are two ways to show that the truncated purifiers still work if $D$ is large enough.
The first one is based on the idea that the quantum walk will not get to the vertex $D$ if the number of iterations is small enough.

\begin{prp}
\label{prp:trunc1}
Consider either the purifier $S$ from \rf{sec:SimpleInfinite} or the purifier $S'$ from \rf{sec:generalInfinite}.
Consider also the algorithm implementing the transduction action of the purifier, either from~\rf{thm:transducer-algorithm} or~\ref{thm:transducer-algorithm-query}.
Then, if $D > 2K$, the action of the algorithm stays the same if we replace the actual purifier with its truncated version at depth $D$.
\end{prp}

\begin{proof}
We consider the transducer $S$, the case of $S'$ being similar.
We also consider the algorithm from \rf{thm:transducer-algorithm} because we gave its proof in \rf{sec:transducer}.
The argument for the second algorithm is similar.

Recall that $S$ acts in the space $\cN$, its public space $\cH$ is spanned by $\ket N|0>$, and its private space $\cL$ by $\ket N|j>$ with integer $j\ge 1$.
We consider the execution of the algorithm from \rf{thm:transducer-algorithm}.
Note that we \emph{do not} apply any perturbations as in the proof of the said theorem.

Let $\zeta_t$ be the content of the register $\cL$ in this algorithm after the $t$-th application of $S$.
We claim that 
\begin{equation}
\label{eqn:zeta}
\zeta_t \in \spn\{\ket N|1>, \dots, \ket N|2t>\}.
\end{equation}

We prove this by induction on $t$.
For the base case, $\zeta_0 = 0$ because the algorithm starts in the state~\rf{eqn:transalgorithm1}.

The algorithm repeatedly couples the content of the register $\cL$ with a fresh copy of $\xi/\sqrt{K}$, where $\xi$ is the initial state.
In the case of the purifier $S$, the initial state $\xi = \ket|0>$.
For the inductive step, assume that~\rf{eqn:zeta} holds for a particular value of $t$.
Consider the next application of $S$.
Since $R_1$ and $R_2$ perform local reflections on the infinite ray in \rf{fig:line} (more formally, by \rf{clm:simpleReflections}), we have that 
\begin{equation}
\label{eqn:zeta2}
R_1\sC[\frac{\ket|0>}{\sqrt K} \oplus \zeta_t] \in \spn\sfigA{\ket N|0>, \dots , \ket N|2t+1> }
\end{equation}
and
\begin{equation}
\label{eqn:zeta3}
S\sC[\frac{\ket|0>}{\sqrt K} \oplus \zeta_t] 
= R_2R_1 \sC[\frac{\ket|0>}{\sqrt K} \oplus \zeta_t] 
\in \spn\sfigA{\ket N|0>, \dots , \ket N|2t+2>}.
\end{equation}
This proves~\rf{eqn:zeta} with $t$ increased by 1.

By \rf{clm:sameAction}, for all $t<D/2-1$,
the action of $R_1$ in~\rf{eqn:zeta2} is identical to the action of $R_1\trunc$, and the action of $R_2$ in~\rf{eqn:zeta3} is identical to that of $R_2\trunc$.
This proves the proposition.
\end{proof}

A similar proof goes through also if the transducer $S$ in \rf{prp:trunc1} is obtained via functional composition as in~\rf{sec:composition} (or any other form of composition from~\cite{belovs:LasVegas}), where some of the constituents are purifiers $S$ or $S'$.
\medskip

The second way of justifying transition to the truncated purifier is based on perturbations as defined in \rf{sec:transPerturbed}.
\rf{thm:main} from the introduction is a corollary of this result.

\begin{thm}
\label{thm:mainFormal}
The statements of \CrossRef{Theorems}{thm:simple} and \ref{thm:mainInfinite} hold after replacing $S$ and $S'$ by their versions truncated at depth $D$ if we allow for a perturbation of size $2(1-\delta)^{D-1}$ in the transducer, and claim that the query complexity is smaller than $1/(2\delta)$ (instead of being exactly equal).
\end{thm}

\begin{proof}
We start with the purifier $S$ from \rf{sec:SimpleInfinite}, and we mimic the proof of \rf{thm:simple} using \rf{clm:sameAction}.
For $p<\frac12$, we have the following version of~\rf{eqn:SimpleSequence1}, with no perturbations
\[
\begin{aligned}
&\sum_{j=0}^{D-1} \gamma^j \ket |j> = \sum_{i=0}^{D/2-1} \gamma^{2i} \sA[\ket |2i> + \gamma \ket |2i+1>]\\
\maps{R_1\trunc} & 
\sum_{j=0}^{D-1} \gamma^j \ket |j>
= \ket |0> + \sum_{i=0}^{D/2-2} \gamma^{2i+1} \sA[\ket |2i+1> + \gamma \ket |2i+2>] + \gamma^{D-1}\ket |D-1>\\
\maps{R_2\trunc}&
\sum_{j=0}^{D-1} \gamma^j \ket |j>.
\end{aligned}
\]

For $p > \frac 12$, we have the following version of~\rf{eqn:SimpleSequence2}:
\[
\begin{aligned}
&\sum_{j=0}^{D-1} (-\gamma)^{-j} \ket |j> 
= \sum_{i=0}^{D/2-1} \gamma^{-2i} \sA[\ket |2i> - \gamma^{-1} \ket |2i+1>]\\
\maps{R_1\trunc} & 
-\sum_{j=0}^{D-1} (-\gamma)^{-j} \ket |j>
= -\ket |0> + \sum_{i=0}^{D/2-2} \gamma^{-2i-1} \sA[\ket |2i+1> - \gamma^{-1} \ket |2i+2>] + \gamma^{-D+1} \ket|D-1>\\
\maps{R_2\trunc}&
-\ket |0> - \sum_{i=0}^{D/2-2} \gamma^{-2i-1} \sA[\ket |2i+1> - \gamma^{-1} \ket |2i+2>] + \gamma^{-D+1} \ket|D-1>\\
\approx&
-\ket |0> - \sum_{i=0}^{D/2-2} \gamma^{-2i-1} \sA[\ket |2i+1> - \gamma^{-1} \ket |2i+2>] - \gamma^{-D+1} \ket|D-1>\\
=& - \ket|0> + \sum_{j=1}^{D-1} (-\gamma)^{-j} \ket |j>,
\end{aligned}
\]
where the ``$\approx$'' represents a perturbation (in the sense of \rf{sec:perturbed}) of size 
\[
2\gamma^{-D+1} 
= 2 \s[\frac{1-p}{p}]^{\frac{D-1}{2}}
= 2 \s[\frac{\frac12-\delta}{\frac12+\delta}]^{\frac{D-1}{2}}
\le 2 \s[1-2\delta]^{\frac{D-1}{2}}
\le 2 (1-\delta)^{D-1}.
\]
The rest of the proof of \rf{thm:simple} stays the same except that the query and transduction complexities in both cases become slightly smaller.

The proof of \rf{thm:mainInfinite} follows the logic of \rf{sec:generalInfinite}.
We prove that~\rf{eqn:1 transduction} holds with $S'$ replaced by $S\ptrunc$ and perturbation $2(1-\delta)^{D-1}$ using a similar reduction as in \rf{lem:one} but to the truncated and perturbed version of the purifier $S$.
Similarly, we show the same for~\rf{eqn:2 transduction}.

In the proof of \rf{thm:mainInfinite}, we use the above two equations to derive~\rf{eqn:mainInfinite1} and~\rf{eqn:mainInfinite2} with the same modifications.
In~\rf{eqn:mainInfiniteCombination}, we use that perturbations in~\rf{eqn:mainInfinite1} and~\rf{eqn:mainInfinite2} act in orthogonal spaces $\cN_0$ and $\cN_1$, to get that the total perturbation is still of size $2(1-\delta)^{D-1}$.
\end{proof}

\subsection{Query-Efficient Implementation}
\label{sec:query-efficient}

As explained in \rf{sec:adversary}, in terms of query complexity, transducers have close connection to the adversary bound for state conversion.
The state conversion problem is usually \emph{finite}.
This means that it is given by pairs $\xi_x\mapsto \tau_x$ in $\cH$ and unitaries $O_x\colon \cM\to\cM$ where $x$ runs through a \emph{finite} set of labels $X$, and the spaces $\cH$ and $\cM$ are \emph{finite-dimensional}.
For instance, one can consider evaluation of a Boolean function, where $x$ runs through the finite set $\cube$ for some $n$.
This motivates the following general result, which can be applied to all transducers in this paper.

\begin{thm}
\label{thm:reducingSpace}
Consider a state conversion problem given by pairs $\xi_x\mapsto \tau_x$ in $\cH$ and input oracles $O_x\colon \cM\to\cM$ as $x$ ranges over a finite set $X$, and $\cH$ and $\cM$ are finite-dimensional.
Let $S=S(O)$ be a transducer with infinite-dimensional private space that solves this problem, i.e., 
$\xi_x \transduce{S(O_x)} \tau_x$ for every $x\in X$.
Then, there exists a canonical transducer $S'=S'(O)$ in a finite-dimensional space such that $\xi_x\transduce{S'(O_x)} \tau_x$ and $L(S',O_x,\xi_x) = L(S, O_x, \xi_x)$ for every $x\in X$.
\end{thm}

\begin{proof}
By \rf{thm:adversary}, the total query states $v_x$ of $S$ give a feasible solution to the adversary bound~\rf{eqn:adversary}.
One complication is that $v_x\in \cH^\uparrow\otimes \cM$ for some infinite-dimensional $\cH^\uparrow$.
Since the space $\cM$ has finite dimension, we have, by Schmidt decomposition, that $v_x \in \cH_x\otimes\cM$ for some finite-dimensional subspace $\cH_x\subset \cH^\uparrow$.
Let $\cH^\downarrow = \spn_{x\in X} \cH_x$.
It is a finite-dimensional subspace of $\cH^\uparrow$ and $v_x\in\cH^\downarrow\otimes \cM$ for all $x\in X$.
By \rf{thm:adversary} again, this implies existence of the transducer $S'$ that uses finite-dimensional space.
\end{proof}

Applying this to \rf{cor:InputOracleErrorReduction}, we get the following result.

\begin{cor}
\label{cor:FinitePurifier}
Let $\cA$ be a qubit, and $\cW$ a finite-dimensional subspace.
Consider a collection of states similar to the ones in~\rf{eqn:InputOracle}
\begin{equation}
\label{eqn:FinitePurifier}
\phi_x = \sqrt{1-p_x} \ket A |0> \ket W|\phi_{x,0}> + \sqrt{p_x} \ket A|1>\ket W|\phi_{x,1}>,
\end{equation}
where $x$ ranges over a finite set $X$, $\phi_{x,0}$ and $\phi_{x,1}$ are normalised, and $p_x\ne 1/2$.
Let, for each $x\in X$, $O_x$ be a unitary such that $\ket|0> \mapsto \ket|\phi_x>$ in $\cA\otimes \cW$.
Then, there exists a finite-dimensional purifier $\purifier$ such that
\[
\purifier(O_x\oplus O_x^*)\colon \ket |0> \transduce{}
\begin{cases}
\ket |0>, &\text{if $p_x < \tfrac 12$;}\\
\ket |1>, &\text{if $p_x > \tfrac 12$}
\end{cases}
\]
with query complexity $1+1/(2\delta_x)$, where $\delta_x = \abs|\frac12-p_x|$.
\end{cor}

Note that the construction in the proof of \rf{thm:reducingSpace} requires restriction to the subspace $\cH^\downarrow$, which means that the corresponding transducer is \emph{not} time-efficient.
On the other hand, it yields \emph{exact} purification.
The new purifier has query complexity $\OO(1/\delta)$, whereas the one from~\cite{belovs:variations} has complexity $\OO(1/\delta^2)$, see \rf{tbl:purifier-comparison}.

This purifier can be used to convert bounded-error algorithms into exact dual adversary in the following way~\cite{belovs:variations}.
Let $A$ be a quantum query algorithm computing a function $f\colon X\to\bool$ with bounded-error, and let $O_x$ be the action of $A$ on the input string $x\in X$.
Then, composing it with the purifier $\purifier$ from \rf{cor:FinitePurifier}, we get a \emph{finite-dimensional} transducer that evaluates the function $f$ \emph{exactly} and whose query complexity is equal to the Las Vegas query complexity of $A$ up to a constant factor (dependent on the gap $\delta$).
The latter transducer also yields the corresponding dual adversary via \rf{thm:adversary}.

\section{Lower Bound}
\label{sec:lower-bound}

In this section, we show that query complexities of the purifiers $S$ and $S'$ from \CrossRef{Sections}{sec:SimpleInfinite} and~\ref{sec:generalInfinite} are optimal.
We prove this via an adversary lower bound for a very simple state conversion problem in the sense of \rf{defn:stateConversion}.
Assume $0<\delta<1/2$ is fixed.
Let $X=\{0,1\}$, $\cM = \bC^2$, and $\cH = \bC$.
Consider two input oracles $O_0, O_1\colon \bC^2\to\bC^2$ reflecting around the states
\begin{equation}
\label{eqn:LBvarphis}
\varphi_0 = \sqrt{\tfrac12 + \delta}\,\ket|0> + \sqrt{\tfrac12 - \delta}\, \ket|1>
\qqand
\varphi_1 = \sqrt{\tfrac12 - \delta}\,\ket|0> + \sqrt{\tfrac12 + \delta}\, \ket|1>,
\end{equation}
respectively.
I.e., $O_i = 2\varphi_i\varphi_i^* - I$.
In notation of \rf{sec:SimpleInfinite}, these are precisely the states $\varphi_{\frac12 \pm \delta}$ from~\rf{eqn:varphi}, and the oracles are $O_{\frac12 \pm \delta}$, but we use 0 and 1 here for conciseness.

Let $\xi_0 = \xi_1 = \tau_0 = \ket |0>$, and $\tau_1 = - \ket|0>$.
This state conversion problem is a restriction of the settings of \rf{thm:simple} for $p = \frac12\pm \delta$, so the purifier of that theorem solves this problem in query complexity $1/(2\delta)$.
This state conversion problem also falls into the settings of \rf{thm:mainInfinite} since the space $\cM$ is 2-dimensional, and the vector $\ket |0>$ belongs to $\spn\sfigA{\ket|0>\ket|\phi_0>, \ket|1>\ket|\phi_1>}$ from the statement of the theorem.
The transducer of the latter theorem also has query complexity $1/(2\delta)$.
The next result shows that this cannot be improved.

\begin{thm}
\label{thm:lowerBound}
For any $0<\delta<1/2$, any transducer $S = S(O)$ solving the above state conversion problem satisfies $\max\sfigB{L(S, O_0, \xi_0), L(S, O_1, \xi_1)} \ge 1/(2\delta)$.
\end{thm}

\begin{proof}
We prove that any feasible solution $v_0, v_1$ to the corresponding dual adversary bound satisfies $\max\sfigA{\|v_0\|^2, \|v_1\|^2} \ge 1/(2\delta)$.
This proves the theorem by \rf{cor:adversary}(a).

\eqrf{eqn:adversaryOriginal} with $x=0$ and $y=1$ read as
\[
\ip<\xi_0, \xi_1> - \ip<\tau_0, \tau_1> = \ipA<v_0, I\otimes \sA[I_\cM - O_0^*O_1]v_1>.
\]
This implies
\begin{equation}
\label{eqn:LB1}
\absB|\ip<\xi_0, \xi_1> - \ip<\tau_0, \tau_1>|\le \norm|v_0|\cdot \norm| I\otimes \sA[I_\cM - O_0^*O_1]|\cdot \norm|v_1|.
\end{equation}
The left-hand side of the above equation is 2, and for the right-hand side, we have
\[
\norm| I\otimes \sA[I_\cM - O_0^*O_1]|
=
\norm| I_\cM - O_0^*O_1|
=
\norm| O_0 - O_1|
=
2\norm|\varphi_0\varphi_0^* - \varphi_1\varphi_1^*|.
\]
Plugging this into~\rf{eqn:LB1}, we get
\begin{equation}
\label{eqn:LB2}
\max\sfigA{\|v_0\|^2, \|v_1\|^2} \ge \|v_0\|\cdot\|v_1\| \ge \frac{1}{\norm|\varphi_0\varphi_0^* - \varphi_1\varphi_1^*|}.
\end{equation}
It is easy to check that
\[
\varphi_0\varphi_0^* - \varphi_1\varphi_1^*
=
\begin{pmatrix}
2\delta & 0 \\ 0 & -2\delta,
\end{pmatrix},
\]
which gives the required lower bound by~\rf{eqn:LB2}.
\end{proof}

\section{Error Reduction using Quantum Signal Processing}
\label{sec:QSP}

In this section, we describe an error reduction algorithm using quantum signal processing.
To the best of our knowledge, this particular version has not been previously stated explicitly, however, very similar constructions have been considered in previous work -- see \cite[Section 3]{gilyen:quantumSingularValueTransformation}.
We were particularly inspired by the discussion in~\cite{rall:phaseEstimation}.
Note that we are using usual Quantum Signal Processing~\cite{low:quantumSignalProcessing} directly here like in~\cite{low:HamiltinianBySpectralAmplification}, and not the more elaborate framework of Quantum Singular Value Transformation~\cite{gilyen:quantumSingularValueTransformation}. 
This streamlined construction uses no additional ancillary qubits, whereas Ref.~\cite{gilyen:quantumSingularValueTransformation} requires a small number, and our construction works with reflecting oracles of the form in \rf{eqn:OrefGeneral}, whereas Ref.~\cite{gilyen:quantumSingularValueTransformation} requires state-generating oracles of the form~\rf{eqn:InputOracle}.

\subsection{Quantum Signal Processing Preliminaries}

Let $x$ and $y$ be real numbers satisfying
\[
-1 \le x,y \le 1
\qqand
x^2 + y^2 = 1.
\]
Define the following unitary transformation%
\footnote{Traditionally, a more complex matrix
$\begin{psmallmatrix} x & \ii y \\ \ii y & x \end{psmallmatrix}$ is used.
The two are equivalent up to a simple change of basis.
}
\begin{equation}
\label{eqn:QSP Signal}
W = W(x,y) = \begin{pmatrix}
x & y \\ y & -x
\end{pmatrix}.
\end{equation}

Let $\alpha = (\alpha_0, \alpha_1, \dots, \alpha_k)$ be a sequence of unimodular complex numbers: $|\alpha_j| = 1$ for all $j$.
Then, the \emph{quantum signal processing} algorithm (QSP) corresponding to $\alpha$ is the following $k$-query quantum algorithm $ U_\alpha = 
U_\alpha(W) $ with an input oracle $W\colon \bC^2\to\bC^2$:
\begin{equation}
\label{eqn:QSP}
U_\alpha (W) = 
\begin{pmatrix}
\alpha_k \\ & -\alpha_k^*
\end{pmatrix}
W
\begin{pmatrix}
\alpha_{k-1} \\ & -\alpha_{k-1}^*
\end{pmatrix}
W
\cdots
W
\begin{pmatrix}
\alpha_{1} \\ & -\alpha_{1}^*
\end{pmatrix}
W
\begin{pmatrix}
\alpha_{0} \\ & -\alpha_{0}^*
\end{pmatrix}.
\end{equation}
Here $\alpha^*$ stands for the complex conjugate of a complex number $\alpha$.

\begin{thm}[\cite{gilyen:quantumSingularValueTransformation}]
\label{thm:QSP}
For any sequence $\alpha = (\alpha_0, \alpha_1, \dots, \alpha_k)$ of unimodular complex numbers, there exist complex polynomials $P(x)$ and $Q(x)$ satisfying the following conditions:
\begin{align}
\deg P &\le k,&  \deg Q &\le k-1, \label{eqn:QSP degree}\\
P(-x) &= (-1)^k P(x),& Q(-x) &= (-1)^{k-1} Q(x), \label{eqn:QSP parity}
\end{align}
%and
\begin{equation}
\label{eqn:QSP condition}
P(x) P^*(x) + Q(x) Q^*(x) (1-x^2) = 1
\end{equation}
and such that
\begin{equation}
\label{eqn:QSP main}
U_\alpha \sA[W(x,y)] = 
\begin{pmatrix}
P(x) & y Q^*(x) \\
y Q(x) & - P^*(x)
\end{pmatrix}.
\end{equation}
And vice versa: for every pair of complex polynomials $P(x)$ and $Q(x)$ and non-negative integer $k$ satisfying~\rf{eqn:QSP degree},~\rf{eqn:QSP parity} and~\rf{eqn:QSP condition}, there exists a sequence $\alpha = (\alpha_0, \alpha_1, \dots, \alpha_k)$ of unimodular complex numbers such that~\rf{eqn:QSP main} holds.
\end{thm}

\begin{thm}[\cite{gilyen:quantumSingularValueTransformation}]
\label{thm:ABC}
Let $R$ be a real polynomial and $k$ a positive integer satisfying
\[
\deg R \le k
,\qquad
R(-x) = (-1)^k R(x),
\]
and
\[
|R(x)| \le 1 \qquad \text{for all\quad $-1\le x\le 1$.}
\]
Then, there exist complex polynomials $P$ and $Q$ satisfying~\rf{eqn:QSP degree}--\rf{eqn:QSP condition} and such that $R(x) = \Re (P(x))$, where $\Re$ stands for the real part.
\end{thm}

\begin{thm}[\cite{low:HamiltinianBySpectralAmplification}]
\label{thm:polynomials}
Let $0 < \eps', \delta' < 1$ be parameters.
There exists an odd real polynomial $R(x)$ of degree $\OO\sA[\frac 1{\delta'} \log \frac 1{\eps'}]$ such that
\begin{alignat}{3}
|R(x)| &\le 1  &&\text{for all}\qquad&  -1 &\le x \le 1; \notag\\ 
R(x) &\ge 1-\eps' &&\text{for all}&  \delta' &\le x \le 1 ;\quad \text{and} \label{eqn:polynomials condition}\\
R(x) &\le -1+\eps'\qquad &&\text{for all}& -1 &\le x \le -\delta' \label{eqn:polynomials condition2} .
\end{alignat}
\end{thm}

\subsection{Proof of \rf{thm:ErrorReduction}}
We will work in the two-dimensional subspace $\cH$ given by the basis vectors $\ket H|0> = \ket A|0>\ket W|\phi_0>$ and $\ket H|1> = \ket A|1> \ket W|\phi_1>$.
By~\rf{eqn:OrefGeneral}, the operator $\Oref$ performs reflection about $\sqrt{1-p} \ket H|0> + \sqrt{p}\ket H|1>$ in this space, and is given by the matrix
\[
\Oref = \begin{pmatrix}
1-2p & 2\sqrt{p(1-p)} \\
2\sqrt{p(1-p)} & 2p-1
\end{pmatrix}.
\]
This matrix is in the form~\rf{eqn:QSP Signal} with $x = 1-2p$ and $y = 2\sqrt{p(1-p)}$.
Thus, $x\ge 2\delta$ if $p\le \frac12-\delta$, and $x \le -2\delta$ if $p\ge \frac12 + \delta$.

Take the real polynomial $R$ of \rf{thm:polynomials} with $\delta' = 2\delta$ and $\eps'$ to be fixed later.
Find the complex polynomials $P$ and $Q$ by \rf{thm:ABC} applied to $R$ and $k=\deg R$.
Plug them into \rf{thm:QSP} to get a sequence $\alpha = (\alpha_0,\dots,\alpha_k)$ such that~\rf{eqn:QSP main} holds.
We apply the Quantum Signal Processing algorithm $U_\alpha$ from~\rf{eqn:QSP} in $\cH$ by replacing $W$ with $\Oref$ and applying the intermediate operations 
$\begin{psmallmatrix} \alpha_{j} \\ & -\alpha_{j}^* \end{psmallmatrix}$ 
to the register $\cA$.
After that, we apply the $Z$ operator to $\cA$.
By \rf{thm:QSP}, the resulting action of this algorithm on $\cH$ is
\[
U(p) = \begin{pmatrix}
P(x) & y Q^*(x) \\
-y Q(x) & P^*(x)
\end{pmatrix}.
\]
Let $P(x) = R(x) + \ii S(x)$ for a real polynomial $S$.
Since the matrix $U(p)$ is unitary (alternatively, by~\rf{eqn:QSP condition}), we have
\begin{equation}
\label{eqn:QSPER Unitarity}
R(x)^2 + S(x)^2 + |yQ(x)|^2 = 1.
\end{equation}

If $p \le \tfrac12 -\delta$, then by~\rf{eqn:polynomials condition}, $1-\eps' \le R(x) \le 1$,
and, by~\rf{eqn:QSPER Unitarity}, $S(x)^2 + |yQ(x)|^2 \le 1- (1-\eps')^2 \le 2\eps'$.
Using that the spectral norm is bounded by the Frobenius norm, we get 
\begin{align*}
\norm |U(p) - I|^2 &= 
{\norm|\begin{pmatrix}
R(x)-1 + \ii S(x) & y Q^*(x) \\
-y Q(x) & R(x)-1 -\ii S(x)
\end{pmatrix}|^2}
\\
& \le 2\sB[ (R(x)-1)^2 + S(x)^2 + \abs|yQ(x)|^2 ] \\&\le 6\eps'.
\end{align*}
Similarly, using~\rf{eqn:polynomials condition2}, we have that if $p\ge \frac12 + \delta$, then
\[
\norm |U(p) + I| \le \sqrt{6\eps'}.
\]
Letting $r=0$ if $p\le \frac12-\delta$ and $r=1$ if $p\le \frac12+\delta$, we obtain
\[
\norm |{ U(p)\ket|\phi> - (-1)^r\ket|\phi> }|\le \sqrt{6\eps'}.
\]
Therefore, we can take $\eps' = \eps^2/6$ to obtain the required precision of the algorithm.

The number of queries and time complexity of the algorithm is $\OO(\deg R)$, which is
$O\sA[\frac1{\delta'} \log \frac1{\eps'}]=
O\sA[\frac1{\delta} \log \frac1{\eps}]$
as required.

\section{Non-Boolean Case}\label{sec:non-Boolean}

In this section, we briefly describe the non-Boolean error reduction case, which we reduce to the Boolean case using the ``Bernstein-Vazirani'' trick from~\cite[Section~4]{jeffery:kDist}.
In the non-Boolean case, we can always assume the output register $\cA = (\bC^2)^{\otimes m}$ is a tensor product of $m$ qubits.
We denote basis elements of $\cA$ by bit-strings in the vector space $\bF_2^m$ over the field $\bF_2=\{0,1\}$.
For $a = (a_1,\dots,a_m)$ and $b = (b_1,\dots,b_m)$ in $\bF_2^m$, we denote by 
\[
a\odot b = a_1b_1 + \cdots + a_mb_m \in \{0,1\}
\]
their inner product, where the operations are in $\bF_2$.
The work register $\reg W$ is still arbitrary.
We assume the input oracle $\Oref$ reflects about some state
\begin{equation}
\label{eqn:NB phi}
\phi = \sum_{a\in \bF_2^m} \sqrt{p_a} \ket A|a> \ket W|\phi_a>
\end{equation}
for normalised $\phi_a$.
Moreover, we assume there exists a (unique) $r$ such that $p_r \ge \frac12 + \delta$.
The task is to perform the transformation
\begin{equation}
\label{eqn:NB ErrorReduction}
\ket B|0> \ket AW |\phi> \maps{} \ket B|r> \ket AW |\phi>,
\end{equation}
where $\reg B$ is a register isomorphic to $\reg A$. 

\begin{thm}
\label{thm:NB}
Let $R$ be a Boolean error reduction subroutine that, given the oracle as in~\rf{eqn:Oref}, performs the transformation in~\rf{eqn:ErrorReductionSign} (possibly with imprecision $\eps$). 
Then, there exists an algorithm that, given an input oracle $\Oref$ reflecting about the state~\rf{eqn:NB phi}, implements the transformation in~\rf{eqn:NB ErrorReduction} (with the same imprecision $\eps$).
The algorithm makes one execution of $R$, uses one additional qubit and $\OO(qm)$ additional elementary operations, where $q$ is the query complexity of $R$.
\end{thm}

\begin{proof}
We use the following transformation, where $C$ is a qubit:
\[
T\colon \ket B|b> \ket C|c> \ket A|a> \mapsto \ket B|b> \ketA C|c\oplus (a\odot b)> \ket A|a>.
\]
It evaluates the inner product of $a$ and $b$ into the register $\cC$, and can be efficiently implemented.
We can decompose $T$ as $T=\bigoplus_{b\in \bF_2^m} T_b$, where $T_b\colon \ket C|c> \ket A|a> \mapsto \ketA C|c\oplus (a\odot b)> \ket A|a> $.

Consider the following state
\begin{multline*}
\phi'_b =
T_b \ket C|0> \ket A\reg W|\phi> 
= \sum_{a\in \bF_2^m} \sqrt{p_a} \ket C|a\odot b> \ket A|a> \ket W|\phi_a>
\\=  \sqrt{1 - p'_b}\; \ket C|0>\ket A\reg W |\phi'_{b,0}> + \sqrt{p'_b}\; \ket C|1> \ket A\reg W |\phi'_{b,1}>,
\end{multline*}
where all $\phi'_{b,c}$ are normalised.
The main idea is that 
\begin{equation}
\label{eqn:NB p'}
\text{
$p'_b \ge \frac 12+ \delta$,\; if $r\odot b = 1$;\qquad
and\qquad 
$p'_b \le \frac12 - \delta,$\; otherwise.}
\end{equation}
Thus, the output of the error reduction subroutine $R$ encodes $r\odot b$ into the phase.
From that, we can obtain $b$ using the Bernstein-Vazirani algorithm~\cite{bernstein:quantumComplexity}.

\myfigure
\label{fig:BV-circuit}
=====
\[
\!\!\!\!
\begin{quantikz}[column sep=0.3cm]
\setwiretype{n}
&               & &[-0.3cm] \text{\Large{$O'$}} \\
\lstick{$\cB$} & \gate[3]{T^*} &\qw      & \qw             & \gate[3]{T}   & \qw\\
\lstick{$\cC$} &               &\gate{Z} & \ctrl[open]{1}  &               & \qw\\
\lstick{$\cA$} &               &\qw      & \gate[2]{\Oref} &               & \qw\\
\lstick{$\cW$} & \qw           &\qw      &                 & \qw           & \qw
\end{quantikz}
\quad\!\!
\begin{quantikz}[column sep=0.3cm]
 \lstick{$\ket B|0>$} &  \gate{H^{\otimes m}}  & \gate[3]{T}\gategroup[4,steps=3,style={dashed}]{} &  \qw & \gate[3]{T^*} & \gate{H^{\otimes m}} & \qw  \rstick{$\ket B|r>$}\\
\lstick{$\ket C |0>$}& \qw &  & \gate[3]{R(O')} & & \qw & \qw  \rstick{$\ket C|0>$}\\
\lstick[2]{$\ket A\reg W|\phi>$} & \qw &  & & & \qw & \qw \rstick[2]{$\ket A\reg W|\phi>$}\\
& \qw & \qw &  & \qw & \qw & 
\end{quantikz}
\]
-----
Implementation of the oracle $O'$ and the algorithm in \rf{thm:NB}.
Assuming $R$ is exact, the operations in the dashed box effectively implement the Bernstein-Vazirani oracle for $r$ the hidden string: $\ket B|b>\ket C|0> \ket AW|\phi> \maps{} (-1)^{r\odot b}\ket B|b>\ket C|0> \ket AW|\phi>$.
=====

Let us fill in the technical details.
We will first assume that $R$ is exact.
It is easy to see that the circuit $O'$ in \rf{fig:BV-circuit} on the left implements the transformation 
\begin{equation}
\label{eqn:O'decomposition}
O' = \bigoplus_{b\in \bF_2^m} O'_b,
\end{equation}
where $O'_b$ reflects about the state $\phi'_b$ in $\cC\otimes\cA\otimes \cW$.
Thus, using parallel composition~\rf{eqn:parallel} over the register $\cB$, where all  composed subroutines are equal to $R$, and taking~\rf{eqn:NB p'} into consideration, we have 
\begin{equation}
\label{eqn:NB 1}
(I_\cB\otimes R)\sA[O'] \colon\quad  \frac 1{\sqrt {2^m}} \bigoplus_{b\in \bF_2^m} \phi'_b \longmapsto \frac 1{\sqrt {2^m}}\bigoplus_{b\in \bF_2^m} (-1)^{r\odot b}\phi'_b,
\end{equation}
which implies
\begin{equation}
\label{eqn:NB 2}
T^*\skB[{(I_\cB\otimes R)\sA[O']}]T\colon
\quad
\frac{1}{\sqrt{2^m}} \sum_{b\in \bF_2^m} \ket B|b> \ket C|0> \ket A\reg W |\phi>
\longmapsto
\frac{1}{\sqrt{2^m}} \sum_{b\in \bF_2^m} (-1)^{r\odot b}\ket B|b> \ket C|0> \ket A\reg W |\phi>.
\end{equation}
Using the same analysis as in the Bernstein-Vazirani algorithm~\cite{bernstein:quantumComplexity}, we get that the algorithm in \rf{fig:BV-circuit} on the right implements the transformation in~\rf{eqn:NB ErrorReduction}.
The subcircuit in the dashed box implements the transformation in~\rf{eqn:NB 2}.

The above analysis assumed that $R$ worked exactly.
If $R$ has imprecision $\eps$, then the whole algorithm also has imprecision $\eps$ because the perturbations in different copies of $R$ in $I_\cB\otimes R$ occur in pairwise orthogonal subspaces.
\end{proof}

This result can be combined with any Boolean error reduction or purification procedure.
Combining with \rf{thm:ErrorReduction}, we obtain the following result.
\begin{cor}
\label{cor:NB ErrorReduction}
Let $\eps, \delta>0$ be parameters.
There exists a quantum algorithm that, given query access to the oracle $\Oref$ reflecting about the state~\rf{eqn:NB phi}, $\eps$-approximately performs the transformation in~\rf{eqn:NB ErrorReduction}.
The algorithm makes $\ell = \OO\sA[\frac1{\delta}\log\frac1{\eps}]$ queries to the oracle $\Oref$, uses one additional qubit, and has time complexity $\OO(\ell + m)$.
\end{cor}

Similarly, we can combine this result with a purifier $S'$ from \rf{thm:mainInfinite} to obtain the following result improving on Theorem 13.3 of~\cite{belovs:taming}.

\begin{cor}
\label{cor:NB Transducer}
Let the registers $\cA$ and $\cW$ be fixed as at the beginning of this section.
There exists a transducer $S'' = S''(\Oref)$ with the following properties.
If the input oracle $\Oref$ reflects about the state $\ket|\phi>$ as in~\rf{eqn:NB phi}, $S''$ performs the transduction
\[
S''(\Oref)\colon \ket B|0> \ket AW |\phi> \transduce{} \ket B|r> \ket AW |\phi>.
\]
The query complexity of the transduction is at most $1/(2\delta)$, where $\delta = p_r-1/2$, its transduction complexity is $\OO(1/\delta)$, and its iteration time complexity is $O(m)$.
\end{cor}

\begin{proof}
Let $O' = O'(\Oref)$ and $A = A(R)$ be the algorithms in \rf{fig:BV-circuit} on the left and on the right, respectively.
We can write the algorithm of \rf{thm:NB} as the functional composition 
\begin{equation}
\label{eqn:NB composition}
A\circ R\circ O'.
\end{equation}
We adhere to the same proof strategy as in \rf{cor:InputOracleErrorReduction}:  We obtain the required transduction action by replacing $R$ by $S'$, and $A$ and $O'$ by the corresponding transducers.

Let $S_O = S_O(\Oref)$ and $S_A = S_A(R)$ be canonical transducers as in \rf{cor:algorithmToTransducer} whose transduction action and query complexity is identical to the action of $O'$ and $A$.
Let $S' = S'(O')$ be the purifier from \rf{thm:mainInfinite} acting in $\cC\otimes \cA\otimes \cW$.
Then, the purifier $S'' = S''(\Oref)$ can be written as the same functional composition as in~\rf{eqn:NB composition}:
\[
S'' = S_A\circ S'\circ S_O.
\]
The correctness of the transducer then follows immediately from the proof of \rf{thm:NB}.

We analyse query complexity of the transduction using \rf{thm:transducerComposition}, the analysis of transduction complexity being similar.
Since $S_A$ and $S_O$ have the same query complexity as $A$ and $O'$, respectively, it suffices to evaluate query complexities of the latter.
The algorithm $O'$ uses just one query, hence,
$
L(S_O, \Oref, \xi) \le 1
$
for every normalised initial state $\xi$.
As $S_A\circ S'$ only accesses $\Oref$ through $O'$, we have
\begin{align*}
L\sA[S_A \circ S' \circ S_O, \Oref, \ket B|0> \ket C|0> \ket AW |\phi>]
&=
L\sB[ S_O, \Oref, {q\sA[S_A \circ S', O', \ket B|0> \ket C|0> \ket AW |\phi>]} ]\\
&\le
L\sA[S_A \circ S', O', \ket B|0> \ket C|0> \ket AW |\phi>],
\end{align*}
where we used~\rf{eqn:linearityScaling} on the second step.
Since $A$ only accesses $O'$ through its only query to $R$ given in~\rf{eqn:NB 1}, we have
\begin{align*}
L\sA[S_A \circ S', O', \ket B|0> \ket C|0> \ket AW |\phi>]
&=
L\s[S',\; O',\; \frac{1}{\sqrt {2^m}} \bigoplus_{b\in\bF_2^m} \phi'_b]\\
&=\frac1{2^m} L\s[S',\; \bigoplus_{b\in \bF_2^m} O'_b,\;  \bigoplus_{b\in\bF_2^m} \phi'_b]
&&\text{using~\rf{eqn:O'decomposition} and~\rf{eqn:linearityScaling}}
\\
&=\frac1{2^m}  \sum_{b\in\bF_2^m} L\s[ S', O'_b, \phi'_b ]
&&\text{by \rf{eqn:parallelTransducerComplexityB}}
\\
&\le \frac1{2\delta}
&&\text{by \rf{thm:mainInfinite} and~\rf{eqn:NB p'}}.\qedhere
\end{align*}
\end{proof}

Time- and query-efficient implementations from \SuperRef Sections\ref{sec:time-efficient} and~\ref{sec:query-efficient} are applicable here as well.
Note that the estimate $1/(2\delta)$ in \rf{cor:NB Transducer} is not tight.
For instance, it is not necessary to execute $R$ for $b=0$, since $r\odot b = 0$ for every $r$.
In addition, the inequalities in~\rf{eqn:NB p'} will be strict for many values of $b$.
We leave investigation of the precise value of the query complexity in the non-Boolean case as future work.
Another interesting problem is error reduction when it is promised that there exists a unique $r$ with somewhat large $p_r$, but whose value is less than $1/2$.

\paragraph{Acknowledgements} AB is supported by the Latvian Quantum Initiative under European Union Recovery and Resilience Facility project no. 2.3.1.1.i.0/1/22/I/CFLA/001.

SJ is supported by NWO Klein project number OCENW.Klein.061;
the project Divide \& Quantum  (with project number 1389.20.241) of the research programme NWA-ORC, which is (partly) financed by the Dutch Research Council (NWO);
and by the European Union (ERC, ASC-Q, 101040624). SJ is a CIFAR Fellow in the Quantum Information Science Program. 

We thank the anonymous referees for comments that have substantially improved the presentation of this manuscript.

\bibliographystyle{habbrvM}
{
\small
\bibliography{belov,refs}
}

\end{document}